\documentclass[11pt]{article}

\usepackage{graphics,graphicx,amssymb,amsmath,amsthm,latexsym,xspace,cases}
\usepackage{cleveref}[2012/02/15]
\usepackage{booktabs}
\usepackage{mathtools}
\usepackage[usenames]{color}

\usepackage{latexsym}
\usepackage{amsfonts}

\newtheorem{theorem}{Theorem}[section]
\newtheorem{lemma}[theorem]{Lemma}

\newtheorem{corollary}[theorem]{Corollary}
\theoremstyle{definition}
\newtheorem{definition}[theorem]{Definition}

\newcommand{\ignore}[1]{}

\newcommand{\etal}{{\em et al.}}

\newcommand{\LABEL}{\mbox{\it Label}}
\newcommand{\EDGE}{\mbox{\it Edge}}
\newcommand{\IND}{\mbox{\it Ind}}

\newcommand{\F}{{\cal F}}

\newcommand{\ceil}[1]{\left\lceil{#1}\right\rceil}

\DeclarePairedDelimiter{\floor}{\lfloor}{\rfloor}

\title{Adjacency labeling schemes and induced-universal graphs}

\author{Stephen Alstrup   \thanks{Department of Computer Science,
University of Copenhagen, Denmark. E-mail:
{\tt s.alstrup@di.ku.dk}.}
\and Haim Kaplan \thanks{Blavatnik School of Computer Science, Tel Aviv University,
  Israel. Research supported by
The Israeli Centers of Research Excellence (I-CORE) program
(Center No. 4/11). E-mail:
{\tt haimk@post.tau.ac.il}.} \and
Mikkel Thorup \thanks{Department of Computer Science,
University of Copenhagen, Denmark. E-mail:
{\tt mikkel2thorup@gmail.com}.} \and  Uri Zwick\thanks{Blavatnik School of Computer Science, Tel Aviv University,
  Israel. Research supported by BSF grant no.\ 2012338 and by
The Israeli Centers of Research Excellence (I-CORE) program
(Center No. 4/11). E-mail: {\tt
	zwick@tau.ac.il}.}}

\begin{document}

% \date{}
\maketitle

\vspace*{-15pt}

\begin{abstract}
\noindent%
% \mtcom{Perhaps nice to say {\em for up to $n$ vertices}. Obviously the
% scheme also works if we only use a subset of the vertices.}
We describe a way of assigning \emph{labels} to the vertices of any undirected graph on up to~$n$ vertices, each composed of $n/2+O(1)$ bits, such that given the labels of two vertices, and no other information regarding the graph, it is possible to decide whether or not the vertices are adjacent in the graph. This is optimal, up to an additive constant, and constitutes the first improvement in almost 50 years of an $n/2+O(\log n)$ bound of Moon. As a consequence, we obtain an \emph{induced-universal} graph for $n$-vertex graphs containing only $O(2^{n/2})$ vertices, which is optimal up to a multiplicative constant, solving an open problem of Vizing from 1968. We obtain similar tight results for directed graphs, tournaments and bipartite graphs.
\end{abstract}

\section{Introduction}

An \emph{adjacency labeling scheme} for a given family of graphs is a way of assigning \emph{labels} to the vertices of each graph from the family such that given the labels of two vertices in the graph, and no other information, it is possible to determine whether or not the vertices are adjacent in the graph. The labels are assumed to be composed of bits and are required to be of the same length. The goal is, of course, to make the labels as short as possible. An adjacency labeling scheme can be used to store a graph \emph{implicitly} in a \emph{distributed} manner. Adjacency labeling schemes first appear in
%  the papers of
% for various families of graphs were considered by
Breuer \cite{Breuer66}, Breuer and Folkman \cite{BF67}, M\"{u}ller \cite{muller}, and Kannan, Naor and Rudich~\cite{KNR92}. (See more references below.)

Various other types of labeling schemes were also considered. In a \emph{distance} labeling scheme, % for example,
given the labels of two vertices it should be possible to deduce the distance between them in the represented graph. In a \emph{routing} scheme, we may want to be able to identify the first edge on a shortest path, or an almost shortest path, between the two vertices. There is a vast literature on these subjects. % We refer the interested reader to Peleg \cite{peleg,peleg2}.
When the graphs considered are rooted trees, we may want to be able to decide whether a vertex is an \emph{ancestor} of another vertex, given just the labels of the two vertices, or to be able to compute the label of their \emph{Nearest Common Ancestor} (NCA).
(See next section and the extensive survey of Gavoille and Peleg~\cite{gavoillepeleg}.)

Closely related to adjacency labeling schemes are \emph{induced-universal graphs}. A graph ${\cal G}=({\cal V},{\cal E})$ is said to be an induced-universal graph for a family $\cal F$ of graphs, if for every graph~$G$ of~$\cal F$ there is an induced subgraph of~$\cal G$ that is isomorphic to~$G$.
% \footnote{We note in passing that \emph{universally inducing} may have been a more appropriate term. We stick with \emph{induced-universal} as it is already well-established.}
Induced-universal graphs were introduced by Rado~\cite{Rado64}.
 % and Moon~\cite{moon1965minimal}.
Kannan \etal~\cite{KNR92} note that a family $\cal F$ has an $L$-bit adjacency labeling scheme if and only if it has an induced-universal graph on at most $2^L$ vertices. Moon~\cite{moon1965minimal} showed that the family of all $n$-vertex undirected graphs has an induced-universal graph on $O(n2^{n/2})$ vertices. To do that, he implicitly constructs an adjacency labeling scheme for $n$-vertex graphs that assigns each vertex an $(\lfloor n/2\rfloor+\lceil\lg n\rceil)$-bit label.\footnote{Throughout the paper, we let $\lg n=\log_2 n$.} Moon~\cite{moon1965minimal} uses a simple counting argument to show that
adjacency labels for $n$-vertex graphs must contain at least $(n-1)/2$ bits, and that any induced-universal graph for $n$-vertex graphs must contain at least $2^{(n-1)/2}$ vertices, showing that his upper bounds are not far from being optimal. Closing the gap between the upper and lower bounds is mentioned as an open problem in Vizing \cite{vizing1968some}.
Bollob{\'a}s and Thomason \cite{bollobas1981graphs} show that a random graph on $\lceil n^2 2^{n/2}\rceil $ vertices is, with high probability, an induced-universal graph for the family of $n$-vertex undirected graphs.
While succinct adjacency labeling schemes and small induced-universal graphs for various families of graphs were subsequently constructed (see the next section for a summary), no progress was made on the most basic problem of finding adjacency labeling schemes and induced-universal graphs for the family of all $n$-vertex graphs.

We obtain an adjacency labeling scheme for $n$-vertex graphs that assigns each vertex an $(\ceil{n/2}+4)$-bit label, which is optimal up to a small additive constant. As a consequence, we also get an induced-universal graph of size $O(2^{n/2})$ which is optimal up to a small multiplicative factor.

Using our techniques we also obtain an $(n+3)$-bit adjacency labeling scheme for $n$-vertex directed graphs, an $(\ceil{n/2}+4)$-bit adjacency labeling scheme for $n$-vertex \emph{tournaments}, thus improving an $(\floor{n/2}+\lceil \lg n\rceil)$-bit  bound of Moon \cite{moon1968topics}, and finally an $(\frac{n}{4}+O(1))$-bit adjacency labeling scheme for $n$-vertex \emph{bipartite} graphs, improving an $(\frac{n}{4}+2\lceil\lg n\rceil)$-bit scheme of Lozin and Rudolf \cite{Lozin2007}. All these results are again optimal up to a small additive constant and give rise to induced-universal graphs that are optimal up to a small multiplicative factor.

\vspace*{-10pt}
\subsection*{The basic challenge}

\vspace*{-5pt}
To illustrate the most basic technical challenge, we briefly consider the simplest
case of \emph{directed} graphs.
% Let us briefly consider adjacency labeling schemes for directed graphs.
Suppose that there is an adjacency labeling scheme that assigns each vertex of an $n$-vertex graph an $L$-bit label. As given the labels of two vertices we can determine whether the vertices are adjacent, the labels of all the vertices determine the graph. As $n(n-1)$ bits are needed to represent a general $n$-vertex directed graph, we get that $L\ge n-1$, i.e., each label must contain at least $n-1$ bits. (For a formal version and a slight strengthening of this argument, see Section~\ref{sec:lower}.)

Suppose now that each vertex~$u$ in an $n$-vertex graph has a distinct index $\,ind(u)\in \{0,1,\ldots,n-1\}$ assigned to it. The graph can then be represented using the adjacency matrix $A=(a_{ij})$, where $a_{ij}=1$ if and only if there is an edge from the vertex whose index is~$i$ to the vertex whose index is~$j$. We can let the label of~$u$ be the $(n-1)$-bit string $adj(u)$ which is simply the $ind(u)$-th row of the adjacency matrix with the diagonal element omitted. Given the labels $adj(u)$ and $adj(v)$ of two vertices~$u$ and~$v$, \emph{and} their indices $ind(u)$ and $ind(v)$, we can easily decide whether there is an edge from~$u$ to~$v$ in the graph. Such an edge exists if and only if $ind(v)<ind(u)$ and $adj(u)[ind(v)]=1$, or $ind(v)>ind(u)$ and $adj(u)[ind(v)-1]=1$. (As can be seen $adj(v)$ is not even required here.)

This labeling scheme seemingly matches the trivial lower bound.
Unfortunately, it is \emph{not} a valid adjacency labeling scheme. To determine whether~$u$ and~$v$ are adjacent, we need to know not only their \emph{labels}, but also their \emph{indices}. (In the sequel, we thus refer to $adj(u)$ as the \emph{tag}, and not the label of~$u$.)

We can of course obtain a valid adjacency labeling scheme by letting the label of a vertex be an encoding of both its index and its tag. But, the resulting labels would then be of length $n+\lceil\lg n\rceil-1$. The fundamental question is whether these extra $\lceil\lg n\rceil$ bits are needed. We show that they are \emph{not} needed. Using a careful choice of indices, we can encode both indices and tags using only $n+O(1)$ bits.

\begin{table}[t]
\renewcommand{\arraystretch}{1.3}
\small
\centering
\makebox[0pt][c]{
\begin{tabular}{|c|c|c|c|}
%\toprule
%\cline{2-5}
\hline

% \vspace*{3pt}
 \hline
\bf Graph family & \bf Lower bound & \bf Upper bound & \bf Reference  \\
\hline
\noalign{\vskip 2mm} \hline
%\midrule
General graphs  &   $2^{\frac{n-1}{2}}$ & $O(n 2^{\frac{n}{2}})$ &  Moon \cite{moon1965minimal} \\
\hline
Tournaments &   $2^{\frac{n-1}{2}}$ & $O(n 2^{\frac{n}{2}})$ & Moon \cite{moon1968topics}  \\
\hline
Bipartite graphs &   $\Omega(2^{\frac{n}{4}})$ & $O(n^2 2^{\frac{n}{4}})$ & Lozin-Rudolf \cite{Lozin2007} % (\cite{alstruprauhe})
\\
\hline \noalign{\vskip 2mm}
\hline
Graphs of max degree $d$, $d$ even & $\Omega(n^{\frac{d}{2}})$ & $O(n^{\frac{d}{2}})$& Butler \cite{Butler_induced-universalgraphs} \\
\hline
Graphs of max degree $d$, $d$ odd & $\Omega(n^{\frac{d}{2}})$ &  $O(n^{{\frac{d+1}{2}}-\frac{1}{d}}\log^{2+\frac{2}{d}}n)$ & Esperet \etal~\cite{Esperet2008} \\
\hline
Graphs of max degree 2&  $\floor{\frac{11n}{6}}$ & $\floor{\frac{5n}{2}} + O(1)$& Esperet \etal~\cite{Esperet2008}    \\

\hline
\noalign{\vskip 2mm} \hline
%Families with an excluded minor
Graphs excluding a fixed minor
 & $\Omega(n)$ & $n^2 (\log {n})^{O(1)} $ & Gavoille-Labourel \cite{gavoille2007shorter}  \\
\hline
Planar graphs & $\Omega(n)$ &$n^2(\log n)^{O(1)}$ & Gavoille-Labourel \cite{gavoille2007shorter} \\
\hline
Planar graphs of bounded degree & $\Omega(n)$ & $O(n^2)$& Chung \cite{Chung90} \\
\hline
Outerplanar graphs & $\Omega(n)$ &  $n (\log n)^{O(1)}$ & Gavoille-Labourel \cite{gavoille2007shorter} \\
\hline
Outerplanar graphs of bounded degree & $\Omega(n)$& $O(n)$& Chung \cite{Chung90}\\
\hline
\noalign{\vskip 2mm} \hline
Graphs of treewidth $k$ & $n2^{\Omega(k)}$& $n (\log \frac{n}{k})^{O(k)} $& Gavoille-Labourel \cite{gavoille2007shorter}\\

\hline
Graphs of arboricity $k$ &  $\frac{n^k}{2^{O(k^2)}}$ & $n^k \min\{(\log n)^{O(1)},2^{O(k\log^*n)}\}$& Alstrup-Rauhe \cite{alstruprauhe} \\
\hline
Forests &  $\Omega(n)$ & $n2^{O(\log^*n)}$& Alstrup-Rauhe \cite{alstruprauhe} \\
\hline
Forests of bounded degree& $\Omega(n)$ & $O(n)$& Chung \cite{Chung90} \\
\hline
Trees of depth $d$ & $\Omega(n)$ & $O(nd^3)$ & Fraigniaud-Korman \cite{fraigniaudkorman2} \\
\hline
Caterpillars & $\Omega(n)$ & $O(n)$& Bonichon \etal~\cite{bonichon} \\

%\bottomrule
\hline
\end{tabular}
}
%\caption{Induced-universal graphs.Above bounded and fixed means the degree assumed being a constant}.
\caption{Induced-universal graphs for various families of graphs. All families considered, except tournaments, are families of undirected graphs.  The results for graphs of maximum degree at most~$d$ assume that~$d$ is a constant.
The $\Omega(n^{d/2})$ lower bound for~$d$ odd is due to Butler~\cite{Butler_induced-universalgraphs}.
In the result for families of graphs with an excluded minor, the~$O(1)$ term in the exponent depends on the fixed minor excluded.}
\label{tab:adjacency2}
\end{table}

\vspace*{-10pt}
\subsection*{Organization of paper}

\vspace*{-5pt}
The rest of this paper is organized as follows. In Section~\ref{sec:summary} we provide a concise summary of related results. In Section~\ref{sec:prelim} we give a formal definition of adjacency labeling schemes and discuss some variants of the definition. In Section~\ref{sec:blocks} we describe the two building blocks used to obtain all our results. The first one of these building blocks, which is the cornerstone of all our constructions, is a labeling scheme for very unbalanced bipartite graphs. The labels produced by this labeling scheme vary drastically in size. Our second building block is a \emph{spreading} scheme used to smooth the differences in the label sizes. Combining the two schemes we manage to assign all vertices labels of the same size, thus conforming to the formal requirement. In Section~\ref{sec:directed} we present our new labeling schemes for \emph{directed} graphs. In Section~\ref{sec:undirected} we present our new labeling schemes for \emph{undirected} graphs. The labeling schemes for directed graphs are presented first as they are somewhat simpler. In Section~\ref{sec:tournaments} we present our schemes for \emph{tournaments}.
%, which essentially follow from our results for undirected graphs.
In Section~\ref{sec:bipartite} we present our results for \emph{bipartite} graphs. The schemes for bipartite graphs require some additional new ideas.
In Section~\ref{sec:efficiency} we discuss the issue of efficient decoding. In Section~\ref{sec:universal} we discuss the construction of induced-universal graphs.  In Section~\ref{sec:lower} we discuss lower bounds. We end in Section~\ref{sec:concl} with some concluding remarks and open problems.

\section{Summary of related results}\label{sec:summary}

A summary of known upper and lower bounds on the size of induced-universal graphs for various families of graphs is given in Table~\ref{tab:adjacency2}.
% The references given there are for the upper bounds.
Corresponding results for adjacency labeling schemes can be obtained by taking logarithms.
We improve the first three upper bounds, making them asymptotically tight.
% The results in the table are mostly self-explanatory.

% Bollob{\'a}s and Thomason \cite{bollobas1981graphs} show that with very high probability, a random graph on $\lceil n^2 2^{n/2}\rceil$ vertices is an induced-universal graph for $n$-vertex graphs.

An induced-universal graph for a family~${\cal F}$ is a graph that contains each graph from~${\cal F}$ as an induced subgraph. A \emph{universal} graph for~${\cal F}$, on the other hand, is a graph that contains each graph from~${\cal F}$ as a subgraph, not necessarily induced. A clique on~$n$ vertices is clearly a universal graph for all $n$-vertex graphs. The challenge is to construct universal graphs with as few edges as possible.
Chung~\cite{Chung90} shows that universal graphs can be used to construct induced-universal graphs. Using  universal graphs constructed by Babai \etal~\cite{BCEGS82}, Bhatt \etal~\cite{BCLR89} and Chung \etal~\cite{CGP76,CG78,CG79,CG83}, she obtains her induced-universal graphs cited in Table~\ref{tab:adjacency2}. The induced-universal graphs for planar graphs, outerplanar graphs, graphs excluding a fixed minor, and bounded degree graphs listed in Table~\ref{tab:adjacency2} also rely on her ideas.
Alon and Capalbo \cite{alon2007sparse,AlonCapalbo2008}, improving many previous results, show that for every fixed~$d$, there is a graph with $O(n^{2-2/d})$ edges which is universal for $n$-vertex graphs of maximum degree at most~$d$, which is asymptotically optimal. Esperet \etal~\cite{Esperet2008} use this result to obtain their induced-universal graphs for graphs of fixed maximum degree~$d$, where~$d$ is odd.

Distance labeling schemes were considered by many authors. See, e.g., Peleg \cite{peleg} and Gavoille \etal~\cite{Gavoille200485} and the references therein. Labeling schemes for flow and connectivity were considered by Katz \etal~\cite{siamcompKatzKKP04} and Korman \cite{Korman2010}.

Labeling schemes for answering ancestor and NCA queries in trees were considered, among others by, Abiteboul \etal~\cite{abiteboul}, Alstrup \etal~\cite{AR02,AGKR04,alstrupnca2014} and Fraigniaud and Korman \cite{fraigniaudkorman}.

Routing schemes were also considered by many authors. See, e.g.,
Eilam \etal~\cite{EilamGP03}, Fraigniaud and Gavoille \cite{Gavoille01}, Thorup and Zwick \cite{ThZw05,throupzwick} and the references therein.

\section{Prelimaries}\label{sec:prelim}

% Let $\G_n$ be the set of graphs on vertex set $[n]=\{0,1,\ldots,n-1\}$.

%\begin{definition}[Adjacency labeling schemes]
%A pair of functions $\LABEL:\G_n\to \{0,1\}^{n\times L}$ and $\EDGE:\{0,1\}^L\times \{0,1\}^L\to\{0,1\}$ is an $L$-bit adjacency labeling scheme for $n$-vertex graphs if and only if for every $G=(V,E)\in \G_n$, where $V=[n]$, and every $i,j\in V$, we have $(i,j)\in E$ if and only if $\EDGE(\LABEL(G,i),\LABEL(G,j))=1$. (Here $\LABEL(G,i)$ denotes the $i$-th row of the $n\times L$ matrix $\LABEL(G)$.)
%\end{definition}
%
%% We refer to $\LABEL(G,v)$ as the label assigned to vertex~$v$ of~$G$. If the graph~$G$ is clear from the context, we write $\LABEL(v)=\LABEL(G,v)$.
%
%\paragraph{How about the following alternative (I have no strong feeling about this, but just wanted to put it out there)}

We begin with a formal definition of adjacency labeling schemes. For concreteness, we assume throughout the paper that every $n$-vertex graph is defined on the vertex set $V=[n]=\{0,1,\ldots,n-1\}$. Every $n$-vertex graph can of course be made a graph on $V=[n]$ by mapping its vertices to~$[n]$.

\begin{definition}[Adjacency labeling schemes]\label{def:label}
Let $\F_n$ be a family of graphs on vertex set $V=[n]=\{0,1,\ldots,n-1\}$.
A pair of functions $\LABEL:\F_n\to \left([n]\to\{0,1\}^{L}\right)$ and $\EDGE:\{0,1\}^L\times \{0,1\}^L\to\{0,1\}$ is an $L$-bit adjacency labeling scheme for $\F_n$ if and only if for every $G=(V,E)\in \F_n$, where $V=[n]$, and every $u,v\in V$, we have $(u,v)\in E$ if and only if $\EDGE(\LABEL(G)(u),\LABEL(G)(v))=1$.
\end{definition}

In Definition~\ref{def:label}, the family $\F_n$ can be a family of \emph{undirected} graphs or of \emph{directed} graphs. If~ $\F_n$ is a family of undirected graphs,
we should of course have $\EDGE(x,y)=\EDGE(y,x)$, for every $x,y\in\{0,1\}^L$.
%The same definition applies, however, to families $\DF_n$ of \emph{directed} graphs.
% In both cases, if the graph~$G$ is understood, we write $\LABEL(u)$ instead of $\LABEL(G)(u)$.

% A few remarks regarding Definition~\ref{def:label} are in order.
Many of the papers on adjacency labeling schemes say that a family $\F_n$ admits an $L$-bit adjacency labeling scheme if and only if given any graph $G\in \F_n$, it is possible to assign each vertex~$u$ of~$G$ an $L$-bit label such that given the labels of two vertices~$u$ and~$v$ it is possible to decide whether they are adjacent in~$G$. It is not difficult to check that this definition is equivalent to our definition. We explicitly refer to the encoding function $\LABEL$, that assigns labels to the vertices of a given graph, and $\EDGE$, the decoding function, that given two labels decides whether the vertices they belong to are adjacent.

An adjacency labeling scheme $(\LABEL,\EDGE)$ for a family $\F_n$ is said to satisfy the \emph{distinctness} property if and only if for every graph $G=(V,E)$ from $\F_n$, and every two distinct vertices $u,v\in V$ we have $\LABEL(G)(u)\ne\LABEL(G)(v)$. Not every labeling scheme satisfies this property. (Of course, if $\LABEL(G)(u)=\LABEL(G)(v)$, then~$u$ and~$v$ must have the same set of neighbors in~$G$.)

Some of the published lower bounds for adjacency labeling schemes rely on the distinctness property. Similar lower bounds can be obtained, however, without relying on it. (See Section~\ref{sec:lower}.) The distinctness property is required if we want to convert a labeling scheme into an induced-universal graph. % All our schemes satisfy the distinctness property.

All our labeling schemes satisfy the distinctness property. Furthermore, for all our labeling schemes it is possible to define an \emph{index} function $\IND:\{0,1\}^L\to [n]$ such that for every graph $G\in \F_n$ and every $u\ne v\in [n]$ we have $\IND(\LABEL(G)(u))\ne \IND(\LABEL(G)(v))$.  However, we would \emph{not} in general have $\IND(\LABEL(G)(u))=u$. Our labeling schemes make an essential use of the freedom to reassign names, i.e.,  indices from $[n]$, to the vertices of the graph. Adjacency labeling schemes that posses such an index function are said to be \emph{indexing}.

% The idea behind the alternative is that we can still say that $\LABEL$
% assigns different labelings to different graphs, where the labeling
% itself assigns labels to the vertices...

% \mtcom{Observations: These labels can be applied to any graph $G=(V,E)$ where
% $|V|\leq n$. We just identify each $v\in V$ with a unique index $i(v)\in [n]$.
% Indices not identifying vertices of $V$ are left just left isolated. There labels will anyway not be
% used.}

If $\F$ is a family of graphs, we let $\F_n$ be the $n$-vertex graphs of~$\F$, and $\F_{\le n}$ the graphs of~$\F$ with at most~$n$ vertices. If every $n'$-vertex graph~$G'$ of~$\F$, where $n'<n$, can be extended into an $n$-vertex graph~$G$ of~$\F$, e.g., by adding $n-n'$ isolated vertices, then a labeling scheme for~$\F_n$, can also be used as a labeling scheme for $\F_{\le n}$. A family $\F$ that satisfies this property is said to satisfy the \emph{extension} property.

When a labeling scheme is used, it is essentially assumed that~$L$, the length of the labels, is known. (Various coding issues arise if~$L$ is not known, or if labels are not of the same length.)
% (We discuss this briefly in Section~\ref{sec:concl}.)
 We may assume that~$n$, the number of vertices in the graph, or an upper bound on this number, is also known. This can be justified as follows. Assume that~$\F$ satisfies the extension property defined above. Let $L_\F(n)$ be the length of the labels assigned by the labeling scheme to the vertices of $n$-vertex graphs of~$\F$. We may assume, without loss of generality, that~$L_\F(n)$ is non-decreasing in~$n$. Given a label size~$L$, we can find the largest~$n$ for which $L_\F(n)=L$ and then infer that the encoded graph has at most~$n$ vertices. The same process should of course be followed when assigning the labels to the vertices.

% Finally, note that Definition~\ref{def:label} implicitly assumes that~$n$, the number of vertices in the graph family under consideration, is known. If~$\F$ satisfies the extension property defined above, then this assumption can be easily replaced by the assumption that~$L$, the label size, is known. Let $L_\F(n)$ be the length of the labels assigned to $n$-vertex graphs of~$\F$. We may assume, without loss of generality, that $L_\F(n)$ is monotonically increasing. Then, given a label size~$L$ we can find the largest $n$ for which $L_\F(n)=L$ and then infer that the encoded graph has at most~$n$ vertices. In the sequel we assume, for simplicity, that~$n$, and hence~$L$, are known.

\section{Building blocks}\label{sec:blocks}

In this section we present our two main new ideas. The new ideas give rise to the two main building blocks used in all our constructions. Both building blocks are labeling schemes for bipartite graphs. They assign each vertex~$u$ both an \emph{index} $\,ind(u)$ and an adjacency \emph{tag} $\,adj(u)$. The pair $(ind(u),adj(u))$ may be viewed as the adjacency label of~$u$. The first scheme needs the freedom to assign indices to the vertices. The second scheme can use indices already assigned to the vertices.

The adjacency tags assigned to the vertices are usually not of the same length. Thus, the resulting labeling schemes do not conform to Definition~\ref{def:label}. They can still be used, however, to construct labeling schemes % for various families of graphs
that do conform to Definition~\ref{def:label}. In a typical application, the graph~$G=(V,E)$ to be encoded is partitioned into~$k$  subgraphs $G_i=(V_i,E_i)$, for $i\in[k]$,
% where the sets $E_i$ are disjoint and
where $E=\cup_{i=1}^k E_i$. Each vertex $u\in U$ is assigned a single index $ind(u)$, used in the encoding of all subgraphs, and a separate adjacency tag $adj_i(u)$ for each subgraph.
(If $u\not\in V_i$, then $adj_i(u)$ is empty.)
The label of~$u$ is then taken to be the tuple $(ind(u),adj_1(u),\ldots,adj_k(u))$. Given $ind(u)$, it would be possible to deduce the length of the tags $adj_1(u),\ldots,adj_k(u)$. While individual tags may have different lengths, the resulting labels would all have the same length.

A bipartite graph $G=(U,V,E)$, where $|U|=k$, $|V|=n-k$, $U\cap V=\emptyset$,
% $U=\{u_0,u_1,\ldots,u_{k_1-1}\}$ and $V=\{v_0,v_1,\ldots,v_{k_2-1}\}$,
and of course $E\subseteq U\times V$, is said to be $(k,n-k)$-bipartite graph.
We usually assume, without loss of generality, that $U=[k]=\{0,1,\ldots,k-1\}$ and $V=[k,n)=\{k,k+1,\ldots,n-1\}$.
Such a bipartite graph can clearly be represented as a $k\times (n-k)$ Boolean adjacency matrix $A=A_G$.

\subsection{A labeling scheme for extremely unbalanced bipartite graphs}\label{sub:unbalanced}

Our main new idea is a labeling scheme for $(k,n-k)$-bipartite graphs $G=(U,V,E)$ where $k\ll n$.
The labeling scheme assigns indices to the vertices of~$V$, thus permuting the columns of the adjacency matrix $A=A_G$, in a way that enables a succinct encoding of the rows of~$A$.

Every $n$-bit string is of the form $0^{t_1}1^{t_2}\ldots$ or $1^{t_1}0^{t_2}\ldots$, where $t_1,t_2,\ldots\ge 1$. Each such maximal block of consecutive 0s or 1s is called a \emph{run}.
If $A=(a_{i,j})$ is a $k\times n$ Boolean matrix and $\pi\in S_n$ is a permutation on $[n]$, we let $A^\pi=(a^\pi_{i,j})$ be the $k\times n$ matrix defined by $a^\pi_{i,j}=a_{i,\pi(j)}$. For convenience, we start the numbering of the rows and columns of~$A$ from~$0$.

\begin{lemma}\label{lem:pi} Let $A$ be an $k\times n$ Boolean matrix. Then, there exists a permutation $\pi\in S_n$ such that the $i$-th row of $A^\pi$ is composed of at most $2^{i}{+}1$ runs. Furthermore, if the $i$-th row is composed of $2^{i}{+}1$ runs, then the first run is a run of 0s. (Recall that row indices start from~$0$.)
\end{lemma}

\begin{proof}
As a warm-up, we begin by proving a slightly weaker statement. We prove that there is a permutation $\pi\in S_n$ for which the $i$-th row of $A^\pi$, for $0\le i<k$, is composed of at most $2^{i+1}$ runs.
We view the columns as binary representations of numbers where the bit in row~$i$ is the $i$-th most  significant bit. For every $j\in \{0,1,\ldots,2^k-1\}$, let $I_j$ be the set of indices of the columns of~$A$ that contain the $k$-bit binary representation of~$j$. Any permutation $\pi$ that sorts the columns in non-decreasing lexicographic order, i.e., places the indices in $I_0$ first, then those of $I_1$, and so on, ending with the indices in $I_{2^k-1}$, satisfies the required condition.

To tighten the bound and obtain the claim of the lemma, we order the blocks $I_0,I_1,\ldots,I_{2^k-1}$ using a \emph{gray code}. The $k$-bit gray code is an ordering of the $k$-bit words such that two consecutive words differ in a single position. For first gray codes are: $\langle 0 , 1\rangle$ and $\langle 00,01,11,10\rangle$. Furthermore, if $\langle g_0,\ldots,g_{2^{b}-1}\rangle$ is the $b$-bit gray code, then $\langle 0g_0,\ldots,0g_{2^{b}-1},1 g_{2^{b}-1},\ldots,1g_{0}\rangle$  is the $(b+1)$-bit gray code. It is easy to verify by induction that the number of times the $i$-th significant bit in a gray code changes is exactly~$2^{i}$. Thus, any permutation $\pi$ that orders the blocks $I_j$ according to a gray code has the property that the $i$-th row in $A^\pi$ is composed of at most $2^{i}+1$ runs. The number of runs may be smaller as some of the index sets $I_j$ may be empty. If the number of runs is exactly $2^{i}+1$, then the first run is a run of 0s.
\end{proof}

\begin{lemma}\label{lem:L} The total number of $n$-bit strings composed of at most $2^i+1$ runs is $R(n,i)=2\sum_{j=0}^{2^i} {n-1 \choose j}$.
Thus, any $n$-bit string composed of at most $2^i+1$ runs can be specified using $L(n,i)=\lceil \lg R(n,i)\rceil$ bits.
\end{lemma}

\begin{proof}
To represent an $n$-bit word composed of $r$ non-empty runs, we need to represent the $r{-}1$ endpoints of the first $r{-}1$ runs. (The first run always starts at position 1, and the $r$-th run always end at position~$n$.) There are thus ${n-1 \choose r-1}$ possibilities. (We have $n-1$ here, as $n$ is the endpoint of the last run, and is therefore not allowed to be the endpoint of any other run.) We need to multiply this number by~2, as the first run may be a run of 0s or a run of 1s.
Summing up we get the desired result.
\end{proof}

Lemma~\ref{lem:pi} states that if the $i$-th row of $A^\pi$ is composed of $2^i+1$ runs, then the first run is a run of 0s. Thus, in the sequel we can actually replace $R(n,i)$ and $L(n,i)$ by $R'(n,i)={n-1\choose 2^i}+2\sum_{j=0}^{2^i-1} {n-1 \choose j}$ and $L'(n,i)=\lceil \lg R'(n,i)\rceil$. This, however, would have only a negligible effect.

%As an immediate corollary we get:
%
%\begin{corollary}\label{C-L1} Any $n$-bit string composed of at most $2^i+1$ runs can be specified using $L(n,i)=\lceil \lg R(n,i)\rceil$ bits.
%\end{corollary}

Let $H(\alpha)=-\alpha\lg \alpha - (1-\alpha)\lg(1-\alpha)$ be the binary entropy function. It is well known that $\sum_{j=0}^k {n \choose j} \le 2^{H(k/n)n}$, for $k\le n/2$. This gives us the following
useful upper bound on~$L(n,i)$.

\begin{lemma}\label{lem:L2} If $2^i\le n/2$, then $L(n,i)\le \lceil H(2^i/n)n \rceil+1$.
\end{lemma}

Using Lemmas~\ref{lem:pi} and~\ref{lem:L} we obtain the following labeling scheme:

\begin{lemma} \label{lem:bipartite}\mbox{\rm [Run encoding]} For every $k\le \lg n$ there is a labeling scheme with the following properties. The scheme receives an $(k,n-k)$-bipartite graph $G=(U,V,E)$, where $|U|=k$ and $|V|=n-k$, with a distinct index $\,ind_1(u)\in[k]$ assigned to every $u\in U$. The scheme assigns a distinct index $\,ind_2(v)\in [n-k]$ to every $v\in V$. It also assigns each vertex~$u\in U$ an $\ell_i$-bit tag $adj_1(u)$, where $i=ind_1(u)$ and $\ell_i = L(n-k,i)\le L(n,i)\le \left\lceil H\left(2^{i}/n\right)n \right\rceil+1$. For every $u\in U$ and $v\in V$, given $(ind_1(u), adj_1(u))$ and~$ind_2(v)$ it is possible to determine whether $(u,v)\in E$.
\end{lemma}

%\addtocounter{theorem}{-1}
%
%\begin{lemma} \label{lem:bipartite}
% There is a labeling scheme for $(k,n-k)$-bipartite on $U=\{u_0,u_1,\ldots,u_{k-1}\}$ and $V=\{v_0,v_1,\ldots,v_{n-k-1}\}$, where $k\le \lg n$, that assigns each vertex of~$V$ a $\lceil\lg n\rceil$-bit label, and vertex~$u_i$, where $i\in[k]$, a $(\lceil\lg n\rceil+\ell_i)$-bit label, where $\ell_i = L(n-k,i)\le L(n,i)\le \left\lceil H\left(2^{i}/n\right)n \right\rceil+1$. The first $\lceil\lg n\rceil$ bits in the labels of the vertices form distinct indices.
%% where $\ell_i=\left\lceil\lg\left({n-1 \choose 2^{i{-}1}} + 2\sum_{j=0}^{2^{i{-}1}-1}{n-1 \choose j}\right)\right\rceil$.
%\end{lemma}

\begin{proof} Let $G=(U,V,E)$ be a bipartite graph.
For every $i\in[k]$, let $u_i\in U$ be such that $ind_1(u_i)=i$.
%, where $U=\{u_0,u_1,\ldots,u_{k-1}\}$ and $V=\{v_0,v_1,\ldots,v_{n-k-1}\}$,
Let $A\in \{0,1\}^{k\times (n-k)}$ be the adjacency matrix of~$G$ in which the $i$-th row corresponds to~$u_i$. The ordering of the columns of~$A$ is arbitrary. Let $\pi\in S_{n-k}$ be a permutation, whose existence follows from Lemma~\ref{lem:pi}, for which the $i$-th row of~$A^\pi$ is composed of at most $2^i+1$ runs.
% The existence of such a permutation follows from Lemma~\ref{lem:pi}.
For every $j\in [n-k]$, let $v_j\in V$ be the vertex whose column is the $j$-th column of~$A^\pi$ and let $ind_2(v_j)=j$.
% (In a sense, $ind_2$ is the inverse of~$\pi$.)

% The label of the $v_j$, for $j\in [n-k]$, is simply an $\lceil\lg n\rceil$-bit encoding of $k+\pi^{-1}(j)\in \{k,\ldots,n-k-1\}$.

The tag $adj_1(u_i)$ is simply an encoding of the
$i$-th row of $A^\pi$, composed of at most $2^{i}{+}1$ runs. By Lemmas~\ref{lem:L} and~\ref{lem:L2}, we can encode this row using $\ell_i=L(n-k,i)\le L(n,i)\le \left\lceil H\left(2^{i}/n\right)n \right\rceil+1$ bits, as required.
(Note that as $i\le k-1$ and $k\le\lg n$, we have $2^{i}\le n/2$, so Lemma~\ref{lem:L2} can indeed be applied.)

If is not difficult to check that, for every $u\in U$ and $v\in V$, given just $ind_1(u),adj_1(u)$ and $ind_2(v)$, it can be determined whether $(u,v)\in E$. Indeed, $ind_1(u)$ tells us which row of the adjacency matrix corresponds to~$u$. Using $ind_1(u)$ and $adj_1(u)$ we can reconstruct this row. The bit in position $ind_2(v)$ then tells us whether $(u,v)\in E$.
%Next we argue that given the labels of two vertices it is possible to determine whether or not they are adjacent. (We assume that both~$n$ and~$k$ are known.) The index of each vertex, which occupies the first $\lceil\lg n\rceil$ bits of its label, allows us to determine whether a vertex belongs to~$U$ or to~$V$. (A vertex belongs to~$U$ if an only if its index is at least $n-k$.)
%%\footnote{At present, there is an even easier way of identifying the vertices of~$V$, as they have $\lceil\lg n\rceil$-bit labels, while the vertices of~$U$ have longer labels.}
%Given the labels of two vertices, we first check whether one of them is in~$U$ and the other is in~$V$. If that is not the case, then the vertices are not adjacent. Otherwise, we identify the vertex that belongs to~$U$. From its index we can infer~$i$, and hence also $\ell_i=L(n-k,i)$. The subsequent $\ell_i$ bits of its label encode the adjacencies of this vertex to all vertices of~$V$, and in particular tell us whether it is adjacent to the second vertex whose label was given.
\end{proof}

% In the proof of Lemma~\ref{lem:bipartite} we did not address the question as to how \emph{efficiently} we can determine whether two vertices are adjacent. We address this question in Section~\ref{sec:efficiency}.

In the present setting, $ind_1(u)$ can be inferred from the length of $adj_1(u)$. However, when the scheme of Lemma~\ref{lem:bipartite} is used as a building block in the construction other labeling schemes, $adj_1(u)$ forms a part of a larger label and $ind_1(u)$ is then used to infer the length of~$adj_1(u)$.

In Section~\ref{sec:efficiency} we consider a modification of the scheme of Lemma~\ref{lem:bipartite} that allows decoding, i.e., determining whether two vertices are adjacent, in constant time, in an appropriate model of computation.

As can be expected, the sum $\sum_{i=0}^{k-1} L(n,i)$ plays an important role in the sequel. As $L(n,i)\le \lceil H(2^i/n)n \rceil + 1$, we get that $\sum_{i=0}^{k-1} L(n,i) \le 2k + \bigl(\sum_{i=0}^{k-1} H(2^i/n)\bigr)n \le 2k + \bar{H}(2^{k-1}/n)n$, where
\[\bar{H}(\alpha) = \sum_{j=0}^{\infty} H\!\left(\frac{\alpha}{2^j}\right)\;.\]
It is not difficult to verify that $\bar{H}(\alpha)$ is well defined, i.e., that the sum converges for any value of~$\alpha$. It is also not difficult to check numerically that
$\bar{H}(\frac{1}{2})=3.15635\ldots$, $\bar{H}(\frac{1}{4})=2.15635\ldots$ and $\bar{H}(\frac{1}{8})=1.34507\ldots$. (Note that as $H(\frac12)=1$, we have $\bar{H}(\frac12)=1+\bar{H}(\frac14)$.)

\subsection{A spreading labeling scheme for bipartite graphs}\label{sub:spread}

We now present a second labeling scheme for $(k,n-k)$-bipartite graphs
% $G=(U,V,E)$
%, where $U=\{u_0,\ldots,u_{k-1}\}$ and $V=\{v_0,\ldots,v_{n-k-1}\}$,
used to counterbalance the labeling scheme of Lemma~\ref{lem:bipartite}.
% $U=[k]=\{0,1,\ldots,k-1\}$ and $V=[n]=\{0,1,\ldots,n-1\}$,
% where of course $E\subseteq U\times V$. We say that such a graph is a $(k_1,k_2)$-bipartite graph.
The labeling scheme receives a bipartite graph $G=(U,V,E)$ with distinct indices $ind_1(u)$, for $u\in U$, and $ind_2(v)$, for $v\in V$, already assigned to its vertices.
The scheme assigns adjacency tags $adj_1(u)$ and $adj_2(v)$ to the vertices $u\in U$ and $v\in V$. The scheme also receives
numbers~$0\le \ell_i\le n-k$, for $i\in [k]$, that control the lengths of the tags assigned to the vertices of~$U$.
The tags of the vertices of~$V$ are all of the same length~$L$, which, of course, depends on the $\ell_i$'s.
The bits contained in the tags $adj_1(u)$ and $adj_2(v)$ are ``raw'' adjacency bits,  no coding tricks are used this time. The scheme only uses the freedom to decide whether the adjacency bit corresponding to a pair $(u,v)\in U\times V$ will reside in $adj_1(u)$ or in $adj_2(v)$. The indices $ind_1(u)$ and $ind_2(v)$ will allow us to determine which of the two tags contains the bit and in which position.
No assumption regarding the relation between~$k$ and~$n$ is required.

\begin{lemma} \label{lem:spread}\mbox{\rm [Spreading]} For every $0\le \ell_i\le n-k$, where $i\in [k]$, there is a labeling scheme with the following properties. The scheme receives an $(k,n-k)$-bipartite graph $G=(U,V,E)$, where $|U|=k$, $|V|=n-k$, with a distinct index $\,ind_1(u)\in[k]$ assigned to every vertex $u\in U$ and a distinct index $\,ind_2(v)\in [n-k]$ assigned to every vertex $v\in V$. The scheme assigns each vertex $u\in U$ an $((n-k)-\ell_i)$-bit tag $adj_1(u)$, where $i=ind_1(u)$. It assigns each vertex $v\in V$ an $L$-bit tag $adj_2(v)$, where $L=\lceil ({\sum_{i=0}^{k-1} \ell_i})/{(n-k)} \rceil$.
For every $u\in U$ and $v\in V$, given $(ind_1(u), adj_1(u))$ and $(ind_2(v),adj_2(v))$, and given the $\ell_i$'s, it is possible to determine whether $(u,v)\in E$.
\end{lemma}

\addtocounter{theorem}{-1}

%\begin{lemma}\label{lem:spread}
%Let $0\le \ell_i\le n-k$, for $i\in [k]$. Then,
%there is a labeling scheme for $(k,n-k)$-bipartite graphs that assigns vertex $u_i$, for $i\in[k]$, a $(\lceil\lg n\rceil+(n-k)-\ell_i)$-bit label, and assigns each vertex $v_j$, for $j\in[n-k]$, an $(\lceil\lg n\rceil+L)$-bit label, where
%$L=\lceil ({\sum_{i=0}^{k-1} \ell_i})/{(n-k)} \rceil$. The first $\lceil\lg n\rceil$ bits in the labels of the vertices are distinct indices that may be chosen arbitrarily.
%%$L=\left\lceil \frac{\sum_{i=0}^{k-1} \ell_i}{k_2} \right\rceil$.
%\end{lemma}

\begin{proof} % Let $u_i=ind_1^{-1}(i)$, for $i\in [k]$, and let $v_j=ind_2^{-1}(j)$.
For every $i\in [k]$, let $u_i\in U$ be the vertex for which $ind_1(u_i)=i$. %Similarly,
For every $j\in[n-k]$, let $v_j\in V$ be the vertex for which $ind_2(v_j)=j$. Let~$A=(a_{i,j})$ be the adjacency matrix of~$G$ in which the $i$-th row corresponds to~$u_i$ and the $j$-th column corresponds to~$v_j$.
We start with each vertex $u_i$, for $i\in [k]$, holding a $(n-k)$-bit tag $adj_1(u_i)$ that specifies its adjacencies to all vertices of~$V$, i.e., the $i$-th row of the adjacency matrix $A$. Each vertex of~$v_j\in V$ starts with an empty tag $adj_2(v_j)$. Our goal is to move~$\ell_i$  bits from $adj_1(u_i)$, for $i\in[k]$, to the tags $adj_2(v_j)$ of some vertices of~$V$ in such a way that each tag $adj_2(v_j)$ will contain roughly the same number of bits. This can be easily done in the following manner. Let $s_0=0$ and $s_i=(\sum_{j=0}^{i-1}\ell_j)\bmod (n-k)$, for $i>0$. We examine the vertices $u_0,u_1,\ldots$ of~$U$ one by one. Vertex $u_i$ removes bit %$a_{i,(s_i+j)\bmod (n-k)}$,
$a_{i,s_i+j}$, for $j\in[\ell_i]$, from its tag and appends it to the tag of vertex~$v_{s_i+j}$.
%$v_{(s_i+j)\bmod (n-k)}$.
In both cases, $s_i+j$ is computed modulo $n-k$.
As the tags of the vertices of~$V$ acquire bits in a round-robin manner, none of them ends up with more than $L=\lceil ({\sum_{i=0}^{k-1} \ell_i})/{(n-k)} \rceil$ bits.

Given the indices and the tags $ind_1(u),adj_1(u)$ and $ind_2(v),adj_2(v)$ of two vertices~$u\in U$ and~$v\in V$, and given all the $\ell_i$'s, it is easy to check whether they are adjacent. Suppose that $i=ind_1(u)$ and $j=ind_2(v)$. If $j$ is not in the (possibly wrapped) interval $[s_i,s_{i+1})$, then the adjacency bit $a_{i,j}$ is contained in~$adj_1(u)$. Otherwise, it is contained in~$adj_2(v)$. Furthermore, the position of~$a_{i,j}$ in~$adj_1(u)$ or $adj_2(v)$ is easily calculated.
If $a_{i,j}$ is in $adj_1(u)$, then it is in position~$j$,
if $j<s_i<s_{i+1}$, in position~$j-\ell_i$, if $s_i<s_{i+1}\le j$, or in position $j-s_{i+1}$, if $s_{i+1}\le j<s_i$. If $a_{i,j}$ is not in $adj_1(u)$, then it is position $\floor{\frac{\bar{s}_i+((j-s_i)\bmod(n-k)}{n-k}}$ of $adj_2(v)$, where $\bar{s}_0=0$ and $\bar{s}_i=\sum_{j=0}^{i-1}\ell_j$, for $i>0$, where the summation this time is not modulo $n-k$. (Note, in particular, that $u$ only needs to know~$\bar{s}_i$ and~$\ell_i$.)
\end{proof}

A slightly improved spreading lemma, used to fine-tune our results, can be found in Appendix~\ref{sec:spread2}.

% \subsection{A slightly improved spreading lemma for bipartite graphs}\label{sub:spread}

\section{Directed graphs}\label{sec:directed}

Let $G=(V,E)$ be a directed graph on $V=[n]$. As we saw in the introduction, the na\"{\i}ve labeling scheme of $n$-vertex directed graphs, without self-loops, assigns to each vertex an $(n+\lceil\lg n\rceil-1)$-bit label.
% This label is composed of a $\lceil\lg n\rceil$-bit index (or id) and the $n-1$ bits of the corresponding row in the adjacency matrix of the graph, omitting the diagonal element which is assumed to be~0.
We provide the first improvement over this na\"{\i}ve bound. Furthermore, our bound is optimal up to a small additive constant.

% The adjacency labeling schemes of this and subsequent sections still assign each vertex~$u$ both an index $\,ind(u)$ and an adjacency information tag $adj(u)$. The label $label(u)$ of a vertex~$u$ is then composed of the concatenation of a binary encoding of $ind(u)$, and $adj(u)$. Given the labels of two vertices, and no other information, it can be determined whether they are adjacent, thus conforming to Definition~\ref{def:label}.

\begin{theorem}\label{thm:directed} For any $n\ge 100$, there is an adjacency labeling scheme for $n$-vertex directed graphs that assigns each vertex an $(n+4)$-bit label.
\end{theorem}

\begin{proof} Let $G=(V,E)$ where $V=[n]$ be a directed graph. Partition the vertex set~$V$ into two sets $A=[k]$, and $B=[k,n)$,
%$\{k,k+1,\ldots,n-1\}$,
where $k=\lceil\lg n\rceil-2$. We can view~$G$ as the disjoint union of~$G[A]$, $G[B]$, $G[A,B]$ and $G[B,A]$, where $G[A]$ and $G[B]$ are the induced directed graphs on~$A$ and $B$, respectively, $G[A,B]=(V,E\cap (A\times B))$ is composed of the edges of~$G$ from~$A$ to~$B$, and $G[B,A]=(V,E\cap (B\times A))$ is composed of the edges of~$G$ from~$B$ to~$A$. The graphs $G[A,B]$ and $G[B,A]$ correspond to  the undirected bipartite graphs $G[A,B]=(A,B,E\cap (A\times B))$ and $G[B,A]=(A,B,E\cap (B\times A))$, obtained by ignoring the direction of the edges.

We start by using the labeling scheme for extremely unbalanced bipartite graphs of Lemma~\ref{lem:bipartite} to represent $G[A,B]$. We assign arbitrary distinct indices to the vertices of~$A$. For concreteness, let $ind_1(i)=i$, for $i\in A$. The scheme of Lemma~\ref{lem:bipartite} assigns indices $ind_2(j)\in [n-k]$ to the vertices of~$B$. It also assigns each vertex $i\in A$ an $\ell_i$-bit tag $adj_1(i)$, where $\ell_i=L(n-k,i)\le L(n,i)\le \lceil H(2^i/n)n \rceil + 1$.

% This labeling scheme assigns vertex~$i$ of~$A$ a label compose of a $\lceil\lg n\rceil$-bit index and additional $\ell_i=L(n-k,i)$ bits. Note that $\ell_i\le L(n,i)\le \lceil H(2^i/n)n \rceil + 1$. It also assigns each vertex of~$B$ a $\lceil\lg n\rceil$-bit index. (The indices of all the vertices in~$A$ and~$B$ are distinct.)

Next, we use the spreading scheme of Lemma~\ref{lem:spread} to represent $G[B,A]$, viewed as a bipartite graph $(A,B,E'')$.
We use the indices $ind_1(i)$ and $ind_2(j)$ assigned to the vertices of~$A$ and~$B$ above.
% As the scheme of Lemma~\ref{lem:spread} can use the indices assigned to the vertices of~$A$ and~$B$ by the scheme used to represent $G[A,B]$, we do not need to assign new indices to the vertices.
We apply Lemma~\ref{lem:spread} with $\ell'_i=(k-1)+\ell_i$, for $i\in [k]$.
As $k=\lceil\lg n\rceil-2$ and $0\le i\le k-1$, we have $\ell_i\le \lceil H(2^{k-1}/n)n\rceil +1 \le \lceil H(1/4)n\rceil+1\le \lceil 0.82n\rceil+1$. Therefore, $\ell'_i\le n-k$, for $i\in [k]$, as required by~Lemma~\ref{lem:spread}. Vertex~$i$ of~$A$ is thus assigned an $((n-k)-\ell'_i)$-bit tag $adj_2(i)$. Each vertex of~$B$ is assigned a $\Delta$-bit tag $adj_3(j)$, where $\Delta = \lceil (\sum_{i=0}^{k-1} ((k-1)+\ell_i))/(n-k) \rceil$.

Next, we use the na\"{\i}ve labeling scheme to encode $G[A]$ and $G[B]$.
We again use the indices~$ind_1(i)$ and~$ind_2(j)$ already assigned to the vertices.
Each vertex $i\in A$ gets a $(k-1)$-bit tag $adj_4(i)$. Each vertex $j\in B$ gets an $((n-k)-1)$-bit tag $adj_5(j)$.

Combing the indices ~$ind_1$ and~$ind_2$ assigned separately to the vertices of~$A$ and $B$, we let
$ind(i)=ind_1(i)$ if $i\in A$, and $ind(j)=k+ind_2(j)$, if $j\in B$. Note that now $ind(u)\in[n]$ for every $u\in V=A\cup B$. For simplicity, we also use $ind(u)$, where $u\in V$, to denote the $\lceil\lg n\rceil$-bit binary encoding of $ind(u)$.

Finally, we assign vertex~$i$ of~$A$ a label composed of the concatenation of $ind(i), adj_1(i), adj_2(i)$ and~$adj_4(i)$, and vertex~$j$ of~$B$ a label composed of the concatenation of $ind(j), adj_3(j)$ and $adj_5(j)$.

Vertex~$i$ of~$A$ is thus assigned a label of length
\[ \lceil\lg n\rceil \;+\; \ell_i  \;+\;  ((n-k)-(k-1)-\ell_i) \;+\; (k-1) \;=\; \lceil\lg n\rceil + (n-k) \;=\; n+2\;. \]
Each vertex of~$B$ is assigned a label of length
\[ \lceil\lg n\rceil \;+\; \Delta \;+\; (n-k-1) \;=\; n+1+\Delta \;. \]
Now,
\[ \Delta \;=\; \left\lceil \frac{\sum_{i=0}^{k-1}((k-1)+\ell_i)}{n-k} \right\rceil \;\le\;
\left\lceil \frac{k(k+1) + n\sum_{i=0}^{k-1} H(2^i/n)}{n-k} \right\rceil \;\le\;
\left\lceil \frac{k(k+1)}{n-k} + \frac{n}{n-k}\bar{H}(2^{k-1}/n) \right\rceil \;.
\]
As $k=\lceil\lg n\rceil-2$, we have $2^{k-1}/n\le\frac{1}{4}$, and thus $\bar{H}(2^{k-1}/n)\le \bar{H}(\frac14)< 2.16$. It is not difficult to verify that for $n\ge 100$ we have $\frac{k(k+1)}{n-k}<0.5$ and $\frac{n}{n-k}\bar{H}(\frac14)<2.5$, and thus $\Delta\le 3$.

The label of each vertex is thus composed of at most $n+4$ bits. We can easily pad the labels of the vertices so that they all contain exactly $n+4$ bits.

Given the labels of two vertices it is possible to determine whether they are adjacent. The index of a vertex, residing in the first $\lceil\lg n\rceil$ bits of its label, tells us whether the vertex is a vertex of~$A$ or of~$B$. It also allows us to break the label into the different tags composing it. Given the indices of two vertices we can easily decide which of the tags to use to determine whether the two vertices are adjacent.
\end{proof}

Theorem~\ref{thm:directed} % and~\ref{thm:directed2}
is also valid for $n<100$, but for that we need to rely on the exact definition of~$L(n-k,i)$ and not just on the convenient upper bounds $L(n-k,i)\le L(n,i)\le \lceil H(2^i/n) \rceil+1$.

The $n+4$ bound of Theorem~\ref{thm:directed} can be improved to $n+3$. When $n$ is a power of~$2$, for example, this is easy. Note that in this case $2^{k-1}/n=\frac18$. As $\bar{H}(\frac18)<1.346$, we get that $\Delta\le 2$. Essentially the same calculation works if $n$ is close, from below, to a power of~$2$, as then $2^{k-1}/n$ is not much larger than $\frac18$. To get the $n+3$ for all sufficiently large values of~$n$, some more work needs to be done. We need to use the slightly more economical way of encoding indices, described in Appendix~\ref{sec:indices}, and the modified spreading lemma of Appendix~\ref{sec:spread2}. The details can be found in Appendix~\ref{sec:directed2}.

The results in this section are for directed graphs without self-loops. Directed graphs with self-loops could of course be handled by adding a single bit to each label.

We defer the treatment of efficient decoding issues to Section~\ref{sec:efficiency}.

\section{Undirected graphs}\label{sec:undirected}

Our scheme for undirected graphs is slightly more complicated than the scheme of directed graphs, as we need to break the graph into more parts. The main ideas, however, are the same.
We start with a simple $(\lfloor n/2\rfloor +\lceil\lg n\rceil)$-bit scheme for $n$-vertex undirected graphs which is implicit in Moon \cite{moon1965minimal}.

\begin{theorem}\label{thm:moon}\mbox{\rm [Moon\cite{moon1965minimal}]} For any $n\ge 1$, there is a labeling scheme
that receives an $n$-vertex undirected graph $G=(V,E)$, with distinct indices $ind(u)\in[n]$ assigned to its vertices, and assigns each vertex an $\lfloor n/2\rfloor$-bit adjacency information tag $adj(u)$. For every two vertices $u,v\in V$, given $(ind(u),adj(u))$ and $(ind(v),adj(v))$ it is possible to determine whether $(u,v)\in E$.
\end{theorem}

%\begin{theorem}\label{thm:moon}[Moon\cite{moon1965minimal}] For any $n\ge 1$, there is a labeling scheme for $n$-vertex undirected graphs that assigns each vertex an $(\lfloor n/2\rfloor +\lceil\lg n\rceil)$-bit label. The first $\lceil\lg n\rceil$ bits in the labels of the vertices are distinct indices that may be chosen arbitrarily.
%\end{theorem}

\begin{proof} Let $u_i\in V$ be the vertex for which $ind(u_i)=i$. Let $A=(a_{i,j})$ be the adjacency matrix of the graph where the $i$-th row and column correspond to~$u_i$. The tag $adj(u_i)$ is composed of the $\lfloor n/2\rfloor$-bit string $a_{i,i+1},a_{i,i+2},\ldots,a_{i,i+\lfloor n/2\rfloor}$, where the addition in the second index is modulo~$n$. This corresponds to arranging the vertices $u_0,u_2,\ldots,u_{n-1}$ in a circle, with each vertex remembering its adjacencies to the $\lfloor n/2\rfloor$ vertices following it in the circle.

Given $(ind(u),adj(u))$ and $(ind(v),adj(v))$ we can easily determine whether $(u,v)\in E$.
% Let $[i,j)$ denote the wrapped interval modulo~$n$, i.e., $[i,j)=\{i,i+1,\ldots,j-1\}$, if $i<j$, and $[i,j)=\{i+1,\ldots,n-1\}\cup \{0,1,\ldots,j-1\}$, otherwise. Suppose that $i=ind(u)$ and $j=ind(v)$.
If $ind(v)-ind(u) \le \lfloor n/2\rfloor$,
the answer is $adj(u)[ind(v)-ind(u)]$; Otherwise, it is $adj(v)[ind(u)-ind(v)]$, where the subtractions $ind(v)-ind(u)$ and $ind(u)-ind(v)$ are interpreted modulo~$n$.
\end{proof}

We note that when~$n$ is even, there is slight redundancy in the scheme just describe, as the adjacency bit $a_{i,i+n/2}$, for every $i\in[n]$, is stored twice. We exploit that later to fine-tune our results.

Theorem~\ref{thm:moon} yields, of course, an $(\floor{n/2}+\lceil\lg n\rceil)$-bit labeling scheme.
Using our techniques, we can reduce the size of the labels to $\lfloor n/2\rfloor +6$.

\begin{theorem}\label{thm:undirected} For any $n\ge 400$, there is a adjacency labeling scheme for $n$-vertex undirected graphs that assigns each vertex an $(\lfloor n/2\rfloor +6)$-bit label.
% The first $\lceil\lg n\rceil$ bits in the labels of the vertices form distinct indices.
\end{theorem}

\begin{proof} Let $G=(V,E)$ be an undirected graph where $V=[n]$. We partition~$V$ into four disjoint sets $A_0,A_1,B_0$ and $B_1$ were $|A_0|=|A_1|=k=\lceil\lg n\rceil - 3$, $|B_0|=\lceil\frac{n}{2}\rceil-k$ and $|B_1|=\lfloor\frac{n}{2}\rfloor-k$. For concreteness, we let $A_0=[0,k)$, $B_0=[k,\lceil\frac{n}{2}\rceil)$, $A_1=[\lceil\frac{n}{2}\rceil,\lceil\frac{n}{2}\rceil+k)$ and $B_1=[\lceil\frac{n}{2}\rceil+k,n)$.
We partition~$G$ into the disjoint union of the four bipartite graphs $G[A_0,B_0], G[A_0,B_1], G[A_1,B_0], G[A_1,B_1]$ and the two undirected graphs $G[A_0\cup A_1]$ and $G[B_0\cup B_1]$.

We assign arbitrary distinct indices to the vertices of $A_0$. For concreteness, we let $ind'(i)=i$, for every $i\in A_0$. Similarly, we let $ind'(i)=i-\lceil\frac{n}{2}\rceil$, for every $i\in A_1$.
We now use Lemma~\ref{lem:bipartite} to encode $G[A_0,B_0]$ and $G[A_1,B_1]$. This assigns distinct indices $ind'(j)\in [\lceil\frac{n}{2}\rceil-k\,]$ to all vertices $j\in B_0$, and distinct indices $ind'(j)\in [\lfloor\frac{n}{2}\rfloor-k\,]$ to all vertices $j\in B_1$. We define distinct indices $ind(u)\in [n]$ to all vertices of~$V$ as follows. If $u\in A_0$, then $ind(u)=ind'(u)$. If $u\in B_0$, then $ind(u)=ind'(u)+k$. If $u\in A_1$, then $ind(u)=ind'(u)+\lceil\frac{n}{2}\rceil$. Finally, if $u\in B_1$, then $ind(u)=ind'(u)+\lceil\frac{n}{2}\rceil+k$.

The labeling scheme of Lemma~\ref{lem:bipartite} also assign
the $i$-th vertices of~$A_0$ and $A_1$ an $\ell_i$-bit tag, where $\ell_i=L(\lceil\frac{n}{2}\rceil-k,i)\le L(\lfloor \frac{n}{2}\rfloor,i)$. (We refrain from explicitly naming the tags.)

To compensate for the $\ell_i$ bits assigned to the $i$-th vertex of~$A_0$ and the $i$-th vertex of~$A_1$, and to leave room for the representation of $G[A_0\cup A_1]$, we use Lemma~\ref{lem:spread} to represent $G[A_0,B_1]$ and $G[A_1,B_0]$, with $\ell'_i=k+\ell_i$, for $i\in [k]$. It is easy to verify that $\ell'_i\le \lfloor\frac{n}{2}\rfloor-k$, for $i\in[k]$, as required by Lemma~\ref{lem:spread}. The $i$-th vertices of~$A_0$ and $A_1$ thus get tags composed of
$(\lceil\frac{n}{2}\rceil-k)-\ell'_i$ bits, and each vertex of $B_0\cup B_1$ gets a tag composed of $\Delta = \lceil (\sum_{i=0}^{k-1} (k+\ell_i)/(\lfloor \frac{n}{2}\rfloor-k) \rceil$ bits. (Tags are padded, if necessary.)
%(For simplicity we refrain from explicitly naming these labels, as we did in the proof of Theorem~\ref{thm:directed}.)

Finally, we use the simple labeling scheme of Theorem~\ref{thm:moon} to represent $G[A_0\cup A_1]$ and $G[B_0\cup B_1]$. We again use the indices already assigned to the vertices. Each vertex of $A_0\cup A_1$ is thus assigned a $k$-bit tag, while each vertex of $B_0\cup B_1$ is assigned a $(\lfloor\frac{n}{2}\rfloor-k)$-bit tag.

As in the proof of Theorem~\ref{thm:directed}, the label assigned to a vertex is the concatenation of the binary representation of its index, and the tags assigned to it for each part of the graph it participates in.
%, in a fixed predetermined order.

The $i$-th vertices of $A_0$ and $A_1$ are thus assigned a label of length
\[ \lceil\lg n\rceil \;+\; \ell_i \;+\; \left(\left(\left\lceil\frac{n}{2}\right\rceil-k\right)-(k+\ell_i)\right) \;+\; k \;=\; \left\lceil\frac{n}{2}\right\rceil+3 \;. \]
Each vertex of $B_0\cup B_1$ is assigned a label of length
\[ \lceil\lg n\rceil \;+\; \Delta \;+\; \left(\left\lfloor\frac{n}{2}\right\rfloor-k\right)
\;=\; \left\lfloor\frac{n}{2}\right\rfloor+3+\Delta \;. \]
Now, as
\[\ell_i \;\le\; L\left(\left\lceil\frac{n}{2}\right\rceil-k,i\right)\;\le\; L\left(\left\lfloor \frac{n}{2} \right\rfloor,i\right) \;\le\; H\left(\frac{2^i}{n/2}\right)\frac{n}{2}+2 \;\le\; H\left(\frac{2^{i+1}}{n}\right)\frac{n}{2}+2\;,\]
we have
\[ \Delta \;=\; \left\lceil \frac{\sum_{i=0}^{k-1}(k+\ell_i)}{\lfloor \frac{n}{2}\rfloor-k} \right\rceil \;\le\;
\left\lceil \frac{k(k+2) + \frac{n}{2}\sum_{i=0}^{k-1} H(2^{i+1}/n)}{\lfloor \frac{n}{2}\rfloor-k} \right\rceil \;\le\;
\left\lceil \frac{k(k+2)}{\lfloor \frac{n}{2}\rfloor-k} + \frac{\frac{n}{2}}{{\lfloor \frac{n}{2}\rfloor-k}}\bar{H}(2^{k}/n) \right\rceil \;.
\]
As $k=\lceil\lg n\rceil-3$, we have $2^{k}/n\le\frac{1}{4}$, and thus $\bar{H}(2^{k}/n)\le \bar{H}(\frac14)< 2.16$. It is not difficult to verify that for $n\ge 400$ we have $\frac{k(k+2)}{\lfloor\frac{n}{2}\rfloor-k}<0.5$ and $\frac{\frac{n}{2}}{\lfloor\frac{n}{2}\rfloor-k}\bar{H}(\frac14)<2.5$, and thus $\Delta\le 3$.

Each vertex is therefore assigned a label of at most $\lfloor\frac{n}{2}\rfloor+6$ bits. Given the labels of two vertices it is possible to decide whether they are adjacent or not.
\end{proof}

A different approach that can be used to prove Theorem~\ref{thm:undirected} is the following.
We partition the vertex set $V=[n]$ into three sets $A,B$ and $C$, where $|A|=k$, $|B|=\lceil\frac{n-k}{2}\rceil$ and $|C|=\lfloor\frac{n-k}{2}\rfloor$. We partition the graph $G=(V,E)$ into $G[A,B],G[A,C],G[B,C],G[A],G[B]$ and $G[C]$. We use recursion to assign indices and tags to~$G[C]$. We use Lemma~\ref{lem:bipartite} to assign indices and tags to $G[A,B]$. Once all indices are assigned, we use Lemma~\ref{lem:spread} to assign tags to $G[A,C]$. We use a simple scheme for balanced bipartite graphs to
assign tags to $G[B,C]$
%that assigns $\lceil\frac{n-k}{4}\rceil$-bit tags to the vertices of $B\cup C$ to represent $G[B,C]$
(see Theorem~\ref{thm:balanced} below). Finally, we use the Moon's scheme (Theorem~\ref{thm:moon}) to assign tags to $G[A]$ and~$G[B]$. The length of the labels produced seems to be essentially the same as those produced in the proof of Theorem~\ref{thm:undirected}.

A improved $(\ceil{n/2}+4)$-bit labeling scheme for $n$-vertex undirected graphs can be found in Appendix~\ref{sec:undirected2}.

\section{Tournaments}\label{sec:tournaments}

A \emph{tournament} is a directed graph $G=(V,E)$ in which every two vertices are connected by an edge in one of the possible directions, i.e., for every $u\ne v\in V$, either $(u,v)\in E$ or $(v,u)\in E$, but not both.

There is a trivial correspondence between tournaments on $V=[n]$ and undirected graphs on $V=[n]$. Given a tournament $G=(V,E)$, we can construct an undirected graph $G'=(V,E')$ where $E'=\{\{u,v\} \mid (u,v)\in E \text{ and } u<v\}$. Conversely, given an undirected graph $G'=(V,E')$, we can construct a tournament $G=(V,E)$ where $E=\{ (u,v) \mid (\{u,v\}\in E' \text{ and } u<v) \text{ or } (\{u,v\}\not\in E' \text{ and } u>v\}$.

It is thus tempting to claim that any labeling scheme for undirected graphs can also be used as a labeling scheme for tournaments, and vice versa. This, however, is not necessarily the case. The problem is that to check whether $u<v$ the vertices need to know their original indices. In our labeling scheme for undirected graphs the labels of the vertices do not retain this information.

However, even though our labeling scheme for undirected graphs assigns new indices to the vertices, it does so in a way that can still be used to represent tournaments. Recall that the labeling schemes partitions~$V$ into four disjoint sets $A_0,A_1,B_0$ and $B_1$. The scheme keeps the original indices of the vertices of~$A_0\cup A_1$ but permutes the indices of the vertices of~$B_0$ and those of~$B_1$. However, these two permutations depend only on $G[A_0,B_0]$ and $G[A_1,B_1]$.

To assign labels to a tournament $G=(V,E)$ on $V=[n]$, we first partition~$V$ into $A_0,A_1,B_0$ and~$B_1$ as done by the labeling scheme for undirected graphs. We assume, without loss of generality, that $A_0=\{0,1,\ldots,k-1\}$, $B_0=\{k,\ldots,\lceil\frac{n}{2}\rceil-1\}, A_1=\{\lceil\frac{n}{2}\rceil+k-1\}$ and $B_1=\{\lceil\frac{n}{2}\rceil+k,\ldots,n-1\}$. We next generate the undirected graph $G'=(V,E')$ corresponding to the tournament~$G$ as above, i.e., $E'=\{(u,v) \mid (u,v)\in E \text{ and } u<v\}$. We now apply the labeling scheme for undirected graphs on $G'[A_0,B_0]\cup G'[A_1,B_1]$. Let $ind(u)$ denote the new index assigned to vertex $u\in V$. We may assume that $ind(u_0)<ind(v_0)<ind(u_1)<ind(v_1)$ for every $u_0\in A_0$, $v_0\in B_0$, $u_1\in A_1$ and $v_1\in B_1$. We now generate a second undirected graph $G''=(V,E'')$, where $E''=\{(u,v) \mid (u,v)\in E \text{ and } ind(u)<ind(v)\}$, and use the scheme for undirected graphs to assign labels to the vertices of~$G''$. It is not difficult to check that the indices assigned to the vertices are the same as those assigned by the first application of the labeling scheme. Thus, given the labels of two vertices in~$G''$ we can determine whether they are adjacent in~$G''$. Using their indices we can then determine the direction of the edge in the tournament~$G$. We thus have:

\begin{theorem}\label{thm:tournament} For any $n\ge 400$, there is an adjacency labeling scheme for $n$-vertex tournaments that assigns each vertex an $(\lfloor n/2\rfloor +6)$-bit label.
% The first $\lceil\lg n\rceil$ bits in the labels of the vertices form distinct indices.
\end{theorem}

The $\lfloor n/2\rfloor +6$ bound can again be improved to $\ceil{n/2}+4$ using the labeling scheme of Theorem~\ref{thm:undirected2}.

\section{Bipartite graphs}\label{sec:bipartite}

In this section we design an almost optimal $(\frac{n}{4}+O(1))$-bit adjacency labeling scheme for bipartite graphs.
% Given the labels of two vertices it can be determined whether the two vertices are adjacent.
In addition to the ideas of the previous sections, a new idea is used to obtain the result.
% (see the proof of Theorem~\ref{thm:bipartite}).

% If~$r$ is known, it is fairly simple to design a scheme that assigns each vertex of a $(\frac{n}{2}-r,\frac{n}{2}+r)$-bipartite graph an $(\lceil\frac{n^2}{4}-\frac{r^2}{n}\rceil+\lceil\lg n\rceil)$-bit label.

The following theorem follows easily form Lemma~\ref{lem:spread} (spreading).
The proof is deferred to Appendix~\ref{sec:bipartite2}.

\begin{theorem}\label{thm:biased} For every~$0\le r<\frac{n}{2}$, there is a labeling scheme for $(\frac{n}{2}-r,\frac{n}{2}+r)$-bipartite graphs, with distinct indices attached to their vertices,
% where $r$ is known,
that assigns each vertex
an $\lceil\frac{n}{4}-\frac{r^2}{n}\rceil$-bit %adjacency information
tag. Given the indices and tags of two vertices, and given~$r$, it is possible to determine whether the two vertices are adjacent.
\end{theorem}

The challenge is again to absorb the $\lceil\lg n\rceil$ index bits, and to do so in a way that works simultaneously for all values of the bias~$r$.
If~$r$ is not known in advance, we can add a $\lceil\lg n\rceil$-bit encoding of it to the labels of the vertices.  (As we only need to reconstruct~$r$ from the labels of two vertices from opposing sides,  $\lceil\frac12\lg n\rceil$ bits are actually enough, but this would not matter.)
% Now, however, we would need to absorb $2\lceil\lg n\rceil$ bits.
If $r\ge \sqrt{2n\lg n}$, then as $\frac{n}{4}-\frac{r^2}{n} < \frac{n}{4}-2\lg n$, we can easily absorb the $2\lceil\lg n\rceil$ bits used to represent~$r$ and the index of each vertex and still obtain labels of size at most $\frac{n}{4}$. As expected, the difficult task is handling bipartite graphs that are almost balanced, i.e., $r<\sqrt{2n\lg n}$.

We begin by designing an adjacency labeling scheme for perfectly \emph{balanced} bipartite graphs. The proof of the following theorem is similar to the proofs of Theorem~\ref{thm:directed} and~\ref{thm:undirected}, though the graph has to be broken into yet more parts. The proof can be found in Appendix~\ref{sec:bipartite2}.

\begin{theorem}\label{thm:balanced} There is a adjacency labeling scheme for $(\frac{n}{2},\frac{n}{2})$-bipartite graphs that assigns each vertex an $(\frac{n}{4}+O(1))$-bit label. The label of each vertex is composed of a distinct index from $[n]$, and an $(\frac{n}{4}-\lg n + O(1))$-bit tag.
% The first $\lceil\lg n\rceil$ bits in the labels of the vertices form distinct indices.
\end{theorem}

To obtain an $(\frac{n}{4}+O(1))$-bit scheme for all bipartite graphs, we design a scheme for almost biased bipartite graphs in which most vertices do not need to know the bias~$r$.
% Using Theorems~\ref{thm:biased} and~\ref{thm:balanced} and a new idea, we obtain:

\begin{theorem}\label{thm:bipartite} There is a adjacency labeling scheme for $n$-vertex bipartite graphs that assigns each vertex an $(\frac{n}{4} +O(1))$-bit label. The label of each vertex is composed of a distinct index from $[n]$, and an $(\frac{n}{4}-\lg n + O(1))$-bit tag.
% The first $\lceil\lg n\rceil$ bits in the labels of the vertices form distinct indices.
\end{theorem}

\begin{proof}
As explained after Theorem~\ref{thm:biased}, there is a simple $(\frac{n}{4}+O(1))$-bit scheme for all $(\frac{n}{2}-r,\frac{n}{2}+r)$-bipartite graphs, where $r\ge\sqrt{2n\lg n}$. We design a new
$(\frac{n}{4}+O(1))$-bit scheme for all $(\frac{n}{2}-r,\frac{n}{2}+r)$-bipartite graphs, where $r<\sqrt{2n\lg n}$. By combining the two schemes, we obtain an $(\frac{n}{4}+O(1))$-bit scheme for all bipartite graphs. (The first bit of each label indicates whether the first or second scheme is used.)

As we have an $O(1)$ term in the statement of the Theorem, and not a specific constant, we allow ourselves to ignore divisibility and integrality issues and avoid the use of ceilings and floors.

Let $R=n^{4/5}$. Let $G=(U,V,E)$ be a $(\frac{n}{2}-r,\frac{n}{2}+r)$-bipartite graph, where $r<\sqrt{2n\lg n}$. Note, in particular, that $r\le \frac{2R^2}{n}=2n^{3/5}$.
% Assume that we have an $(n/2-r,n/2+r)$ bipartite graph. If $r>2R^2/n=2n^{3/5}$, then we already know how to solve the problem. Assume, therefore, that $r\le 2R^2/n=2n^{3/5}$.
Partition~$U$ into a set~$U_0$ of size $\frac{n}{2}-R$ and a set~$U_1$ of size $R-r$. Similarly, partition~$V$ into a set~$V_0$ of size $\frac{n}{2}-R$ and a set~$V_1$ of size $R+r$. We view the vertices of $U_0$ and $V_0$ as ordinary, and the vertices of~$U_1$ and~$V_1$ as special. The graph~$G$ is thus partitioned into the disjoint union of the four bipartite graphs $G[U_0,V_0], G[U_0,V_1],G[U_1,V_0]$ and $G[U_1,V_1]$. The main idea is to assign the ordinary vertices of $U_0\cup V_0$ labels that do not depend on~$r$. The labels of the special vertices of $U_1\cup V_1$ would contain an encoding of~$r$, but as they form only a negligible fraction of all vertices, this could be `smoothed' out.

We start by encoding $G[U_0,V_0]$ using the scheme of Theorem~\ref{thm:balanced}. Each vertex of $U_0\cup V_0$ gets a distinct index in $[n-2R]$ and an $(\frac{n}{4}-\frac{R}{2}+O(1))$-bit tag. (The label of each vertex includes an encoding of its index.) We assign the vertices of $U_1\cup V_1$ distinct indices from~$[n-2R,n)$.

% We still need to encode the subgraphs induced by $A_1\cup B_0$, $A_0\cup B_1$ and $A_1\cup B_1$.

We next use the spreading technique of Lemma~\ref{lem:spread} to encode $G[U_1,V_0]$. We find it more informative to redo the relevant calculations here.
% We use $\ell_i=\frac{n}{2}-\frac{3R}{2}$, for every $i\in [R-r]$.
We need to split the $(R-r)(\frac{n}{2}-R)$ bits describing the adjacencies in $G[U_1,V_0]$ between the vertices of~$U_1$ and~$V_0$.
% Initially we place all these bits in the vertices of~$A_1$, with each vertex of~$A_1$ holding $\frac{n}{2}-R$ bits.
As the tag of each vertex of $V_0$ is already of size $\frac{n}{4}-\frac{R}{2}+O(1)$, and as we want the tag of each vertex of~$V_0$ to be of size $\frac{n}{4}+O(1)$, each vertex of~$V_0$ gets~$\frac{R}{2}$ of these bits. (As $|U_1|=R-r$, this corresponds to applying Lemma~\ref{lem:spread} with $\ell_i=\frac{R}{2}-r$, for every $i\in [\frac{n}{2}-R]$, on $G[V_0,U_1]$. Note that the sides here are reversed.) The number of bits each vertex of~$U_1$ receives is thus
\[ a \;=\; \frac{(R-r)(\frac{n}{2}-R)-\frac{R}{2}(\frac{n}{2}-R)}{R-r} \;=\;
\frac{(\frac{n}{2}-R)(\frac{R}{2}-r)}{R-r}\;. \]
(Note that~$a$ corresponds to~$L$ of Lemma~\ref{lem:spread}.)
The $\frac{R}{2}$ bits that each vertex of~$V_0$ gets are appended to its tag. Vertices of $V_0$ do not  know the meaning of these bits, as they do not know~$r$, but the vertices of~$U_1$ do, as they will know~$r$.

Similarly, each vertex of $U_0$ gets $\frac{R}{2}$ additional bits, and the number of bits left for each vertex of~$V_1$ is
\[ b \;=\; \frac{(R+r)(\frac{n}{2}-R)-\frac{R}{2}(\frac{n}{2}-R)}{R+r} \;=\;
\frac{(\frac{n}{2}-R)(\frac{R}{2}+r)}{R+r}\;. \]

Next, we verify that $b\le \frac{n}{4}$ if and only if $r\le\frac{2R^2}{n-4R}$. As we assumed that $r\le \frac{2R^2}{n}<\frac{2R^2}{n-4R}$, this condition is satisfied. It can also verified that $a\le \frac{n}{4}$ for every $r<R$. (To see this check that if $r=0$, then $a=\frac{n}{4}-\frac{R}{2}$, and that~$a$ is a decreasing function of~$r$ for $0\le r<R$, as the derivative of~$a$ is terms of~$r$ is $-\frac{R(\frac{n}{2}-R)}{2(R-r)^2}$.)

We still need to represent $G[U_1,V_1]$ by splitting the corresponding adjacency bits between the vertices of~$U_1$ and~$V_1$. We again use the spreading technique of Lemma~\ref{lem:spread}. Overall, there are $(R-r)(R+r)=R^2-r^2$ such adjacency bits. We need to verify that we can accommodate them without
any vertex of~$U_1$ and $V_1$ getting more than $\frac{n}{4}$ bits overall. A simple `volume' argument can be used to show that we still have enough space in the tags of the vertices of~$U_1$ and $V_1$. More specifically, we know that all adjacencies between $U_0\cup U_1$ and $V_0\cup V_1$ can be encoded using at most~$\frac{n}{4}$ bits per vertex. As each vertex of $U_0$ and $V_0$ already has $\frac{n}{4}$ bits, and as all adjacencies between $U_0$ and $V_0$, $U_0$ and $V_1$, and~$U_1$ and $V_0$ were encoded, there is enough room left in the tags of~$U_1$ and $V_1$ to encode the adjacencies between these two sets. We can also verify it using a simple direct calculation. The total number of bits currently used by vertices of~$U_1$ and $V_1$ is $(R-r)a+(R+r)b=(\frac{n}{2}-R)R$. The total capacity of these vertices is $2R\cdot \frac{n}{4}=\frac{Rn}{2}$, and $\frac{Rn}{2}-(\frac{n}{2}-R)R = R^2 > R^2-r^2$. Thus, there is indeed enough space.

One problem still remains. The label of each vertex of $U_1\cup V_1$ should also contain $2\lg n$ bits specifying the index of the vertex and~$r$. Thus, while the labels of all vertices of $U_0\cup V_0$ are all of size $\frac{n}{4}+O(1)$, the labels of the vertices of $U_1\cup V_1$ are currently of size $\frac{n}{4}+2\lg n+O(1)$. This can be easily fixed, however, by persuading each vertex of $U_0$ and $V_0$ to hold one more adjacency bit to $V_1$ and $U_1$, respectively. The number of bits in the labels of~$U_1$ and $V_1$ decreases by $\frac{(\frac{n}{2}-R)}{R+r}\gg 2\lg n$, leaving more than enough room in the label of each vertex to store its index and~$r$.

Finally, given the labels of two vertices, it can be determined whether they are adjacent.
\end{proof}

\section{Efficient decoding}\label{sec:efficiency}

In this section we show that the schemes of the preceding sections could be modified so that two vertices need to exchange only $O(\lg n)$ bits of information between them, in a constant number of communication rounds, and spend only $O(1)$ computation time, to decide whether they are adjacent or not. For concreteness, we consider the case of directed graphs. The same ideas apply to all our schemes.

Note that this is easily achieved using the simple $(n+\lceil\lg n\rceil-1)$-bit scheme.
Consider a distributed setting in which each vertex of the graph is a RAM machine. The label of each vertex is stored in its internal random access memory, assumed to be composed of $w$-bit words, where $w\ge\lceil\lg n\rceil$. In particular, the index of a vertex resides in the first word used to represent its label. In the simple $(n+\lceil\lg n\rceil-1)$-bit scheme, to determine whether there is an edge from~$u$ and~$v$, $v$ sends to~$u$ its $\lceil\lg n\rceil$-bit index. Vertex~$u$ can then access the appropriate adjacency bit in its tag in $O(1)$ time. Our goal is to show that something similar could also be done using our schemes. (Note that when labels are stored in $\lceil\lg n\rceil$-bit words, our improved schemes usually save one  memory word.)

To decode our $(n+O(1))$-bit scheme in $O(1)$ time, we need to overcome two obstacles. First, we need to be able to decode the succinct run length encoding used in Lemma~\ref{lem:bipartite} is constant time. Second, we need to be able to keep track, in constant time, of the bit movements performed by the spreading lemma (Lemma~\ref{lem:spread}). To solve the first problem we use the following result.

\begin{theorem}\mbox{\rm [P{\v a}tra{\c s}cu \cite{patrascu08succinct}]}\label{thm:Patrascu}
On a RAM with $\Omega(\lg n)$-bit words, a Boolean array
$A[0\ldots n-1]$ containing $k$ ones and $n-k$ zeros can be represented using $\lg {n \choose k
} + \frac{n}{\lg^t(n / t)} + \tilde O(n^{3/4})$ bits of memory, supporting
  \emph{rank} and \emph{select} queries in $O(t)$ time.
\end{theorem}

A $rank(i)$ query, where $i\in[n]$, asks for the number of 1s in $A[0\ldots i]$.
A $select(i)$ query requests the index of the $i$-th 1 in the array. We only need $rank$ queries. Theorem~\ref{thm:Patrascu} assumes that the number of 1s in the array is exactly~$k$. However, it is not difficult to extend the result for the case in which the array contains at most~$k$ 1s.
Perhaps the simplest way of doing it is to add $\lceil\lg n\rceil$ bits, which are absorbed in the $\tilde O(n^{3/4})$ term, to encode the actual number of 1s.
%One way of doing it is to add~$k$ bits to the array and use them to make sure that the number of 1s is exactly~$k$.

As we saw in the proof of Lemma~\ref{lem:L}, we can represent an $n$-bit string by its first bit and the end positions of its runs. Thus, we can represent an $n$-bit string composed of at most~$r$ runs using its first bit and an $n$-bit string containing at most~$r$ 1s. The first bit of the string and the parity of $rank(i)$ would then tell us whether the $i$-th bit of the string is a~0 or a~1.

Note that the $\lg {n \choose k} $ term in Theorem~\ref{thm:Patrascu} is the information theoretic lower bound, which essentially corresponds to our function $L(n,i)$, when $k=2^i$. The price paid for the efficient decoding is the additive ${n}/{\lg^t(n / t)} + \tilde O(n^{3/4})$ term. If we use $t=2$, then the number of bits lost is only $O(n/\lg^2 n)$. We need to encode about $\lg n$ sparse arrays, with the $i$-th one of them containing at most $2^i$ 1s. Thus the total number of bits lost in all these encodings is only $O(n/\lg n)$. We can easily compensate for these $O(n/\lg n)$ additional bits by slightly adjusting the parameters used in the application of the spreading lemma. (More specifically, we let $\ell_i = \lg {n \choose 2^i} + \frac{n}{\lg^2 n} + \tilde O(n^{3/4})$, instead of $\ell_i=L(n,i)$.) As the $O(n/\lg n)$ additional bits are spread over almost~$n$ tags, each tag acquires at most one additional bit.

%We can actually go one step further. As the new $\ell_i = \lg {n \choose 2^i} + {n}/{\lg^2 n} + \tilde O(n^{3/4})$ are dominated, for small values of~$i$, by the ${n}/{\lg^2 n}$ term, we can actually let all the $\ell_i$'s be \emph{equal}. We choose~$k$ such that $\lg {n \choose 2^k}\le {n}/{\lg^2 n}$, i.e., $k=\lg n-\Theta(\lg\lg n)$. We can then take $\ell_i=\ell = 3n/\lg^2 n$, for all $i\in [k]$. The fact that~$k$ was reduced from~$\lg n-\Theta(1)$ to $k=\lg n-\Theta(\lg\lg n)$ again has only a negligible effect.

% Having all $\ell_i$'s equal solves our second problem. When $\ell_i=\ell$, for every $i\in [k]$, the movements performed in the spreading process of Lemma~\ref{lem:spread} are regular, and we can easily determine in constant time the location of each adjacency bit. (Note, in particular, that in the proof of Lemma~\ref{lem:spread} we simply have $s_i = i\ell\bmod (n-k)$, for $i\in[k]$.)

We next consider the efficient decoding of tags produced using the spreading lemma (Lemma~\ref{lem:spread}).
We use the spreading lemma in two different ways. In some applications, all the $\ell_i$'s are equal. In others, the $\ell_i$'s differ, but $k\le\lg n$. If $\ell_i=\ell$, for every $i\in [k]$, the bit movements performed are regular, and we can easily determine in constant time the location of each adjacency bit. (Note, in particular, that in the proof of Lemma~\ref{lem:spread} we simply have $\bar{s}_i = i\ell$.) Also, $\ell$ can be deduced from the label. In the other case, we simply add an encodings of $\bar{s}_i$ and $\ell_i$ to the appropriate labels. The extra $2\lg n$ bits added are again absorbed in the $ \tilde O(n^{3/4})$ term of the $k\le\lg n$ corresponding vertices. The decoding can then again be made in constant time.

%\section{Oriented graphs}
%I like Uri's observation that this is like tournaments, but with bits instead of trits, with an extra option that an edge is not present. The bit arrays
%will end up trit arrays, which can be represented and accessed
%perfectly \cite{DPT10}.
%
%And with the special vertices and runs, we just get numbers reprsented
%ternary.

\section{Induced-universal graphs}\label{sec:universal}

As observed by Kannan \etal~\cite{KNR92}, an $L$-bit adjacency labeling scheme for a family $\F_n$ yields immediately a $2^L$-vertex induced-universal graph for~$\F_n$. Thus, using Theorem~\ref{thm:undirected} we obtain, in particular, an induced-universal graph for $n$-vertex undirected graphs containing only $O(2^{n/2})$ vertices, resolving the open problem of Moon~\cite{moon1965minimal} and Vizing~\cite{vizing1968some}.

\section{Lower bounds}\label{sec:lower}

Previous lower bounds on the label sizes
% required to represent directed, undirected and bipartite graphs
assume that labels of different vertices are distinct. We increase the lower bounds by~$1$ without relying on this assumption. For \emph{indexing} adjacency labeling schemes, we increase the lower bounds by~$2$ .
Our basic lower bounds follow from the following obvious lemma.

%\begin{corollary} \label{easylower}
%Let $\mathcal{F}$ be a family of $n$-vertex graphs. The labels of any labeling scheme for graphs from $\mathcal{F}$ must be of size at least $(\log |\mathcal{F}|)/n$ bits, where $|\mathcal{F}|$ here taken to be the number of graphs from $\mathcal{F}$ with distinct vertex numbers from $\{1,2,..,n\}$ assigned to them.
%\end{corollary}

% The results in this section are not new and are included for completeness. The proofs we give are slightly simpler. Also, unlike previous proofs, they do not rely on labels being distinct.

% We begin with the following simple lemma from which all our lower bounds follow.
% Essentially all known lower bounds for adjacency labeling schemes follow from the following simple observation.

\begin{lemma}\label{lem:injective} If $(\LABEL,\EDGE)$ is  an adjacency labeling scheme for~$\F_n$, then $\LABEL$ is injective, i.e.,
for every $G\ne G'\in \F_n$ we have $\LABEL(G)\ne\LABEL(G')$.
\end{lemma}

\begin{proof} Let $G=(V,E),G'=(V,E')\in\F_n$. If $\LABEL(G)=\LABEL(G')$, then for every $u,v\in V$ we have
\[\EDGE(\LABEL(G)(u),\LABEL(G)(v)) \;=\; \EDGE(\LABEL(G')(u),\LABEL(G')(v))\;.\] Hence $(u,v)\in E$ if and only if $(u,v)\in E'$ and thus $G=G'$.
\end{proof}

\begin{theorem}\label{thm:lower1} If there is an $L$-bit adjacency labeling scheme for $\F_n$, then $L> \frac{1}{n}\lg|\F_n|$.
% If the labeling scheme assigns distinct labels to distinct vertices, then $L> (\lg|\F_n|)/n$.
\end{theorem}

\begin{proof} Suppose that $(\LABEL,\EDGE)$ is a labeling scheme for $\F_n$.
By Lemma~\ref{lem:injective}, $\LABEL$ is injective and thus $|\F_n|\le 2^{nL}$. This immediately implies that $L\ge \frac{1}{n}\lg|\F_n|$. To show that the inequality is strict, we need to show that there is at least one ordered tuple of labels that cannot be produced by $\LABEL$. Consider the $2^L$ tuples composed of~$n$ identical labels. Each such tuple may only correspond to the empty graph on~$n$ vertices or to the clique on~$n$ vertices. Thus, at least $2^L-2$ of these tuples are not produced by the labeling scheme.
Hence $|\F_n|< 2^{nL}$ and thus $L> \frac{1}{n}\lg|\F_n|$.
% If $\LABEL$ assigns distinct labels, then $|\F_n| \le 2^L (2^L-1)\cdots (2^L-n-1)<2^{nL}$, and hence $L> (\lg|\F_n|)/n$.
\end{proof}

Note that in Theorem~\ref{thm:lower1}, $|\F_n|$ denotes the number of \emph{named} graphs from~$\F_n$, i.e., graphs of $\F_n$ on~$[n]$. Graphs with different names are considered different even if they are isomorphic.

We let $\overline{\F}_n$ be the set of isomorphism classes of graphs from~$\F_n$. If the labeling scheme satisfies the distinctness assumption, then the condition $|\F_n|< 2^{nL}$ used in the proof of Theorem~\ref{thm:lower1} can be replaced by the slightly stronger inequality $|\overline{\F}_n|\le {2^L \choose n}$. (See. e.g., Alstrup and Rauhe \cite{alstruprauhe}.) (To see that this is a slightly stronger inequality, note that $\frac{|\F_n|}{n!} \le |\overline{\F}_n| \le {2^L \choose n} < \frac{2^{nL}}{n!}$.)
However, as $L$ is an integer, the resulting lower bound on~$L$ is usually the same, even though a stronger assumption is made. We note in passing that, without relying on the distinctness assumption, we can get $|\overline{\F}_n|< \bigl(\!{2^L \choose n}\!\bigr)$, where $\bigl(\!{2^L \choose n}\!\bigr)={2^L+n-1 \choose k}$ is the number of multi-subsets of $[2^L]$ of size~$n$.

In the proof of Theorem~\ref{thm:lower1}, we viewed $\LABEL(G)$ as the ordered tuple $(\LABEL(G)(0),\LABEL(G)(1),\ldots,\allowbreak\LABEL(G)(n{-}1))$. We let $\overline{\LABEL}(G)$ denote the corresponding (multi-)set $\{\LABEL(G)(0),\LABEL(G)(1),\ldots,\allowbreak\LABEL(G)(n-1)\}$ in which the order of the labels is ignored. Analogous to Lemma~\ref{lem:injective}, we have the following lemma whose simple proof if omitted.

\begin{lemma}\label{lem:injective2} If $(\LABEL,\EDGE)$ is  an adjacency labeling scheme for $\F_n$, then
% $\LABEL$ is injective, i.e.,
for every $G,G'\in \F_n$, if $G$ and~$G'$ are not isomorphic, then $\overline{\LABEL(G)}\ne\overline{\LABEL}(G')$.
\end{lemma}

Relying on Lemma~\ref{lem:injective2}, we get our second lower bound.

\newcommand{\LL}{{\cal L}}

\begin{theorem}\label{thm:lower2} If there is an \emph{indexing} $L$-bit adjacency labeling scheme for $\F_n$, then $L\ge  \frac{1}{n}\lg|\F_n| + \frac{1}{n}\lg\frac{n^n}{n!}$. For $n\ge 200$, we have $L> \frac{1}{n}\lg|\F_n| + 1.4$
% If the labeling scheme assigns distinct labels to distinct vertices, then $L> (\lg|\F_n|)/n$.
\end{theorem}

\begin{proof} Suppose that $(\LABEL,\EDGE)$ is an indexing labeling scheme for $\F_n$ and let $\IND$ be an appropriate index function. Let $\LL_i=\IND^{-1}(i)$, for $i\in [n]$. Note that $\sum_{i=0}^{n-1}|\LL_i|=2^L$. For every graph $G\in \F_n$, we have $|\overline{\LABEL}(G)\cap \LL_i|=1$, for $i\in [n]$. Thus, the number of sets of labels is at most $\prod_{i=0}^{n-1}|\LL_i|$.
By Lemma~\ref{lem:injective2}, two non-isomorphic graphs must have distinct label sets.
%, we get that
% $|\overline{\F}_n|\le \prod_{i=0}^{n-1}|\LL_i|$. As $\sum_{i=0}^{n-1}|\LL_i|=2^L$, we get that $\prod_{i=0}^{n-1}|\LL_i|  \le ({2^L}/{n})^n$. As $|\F_n|/n! \le |\overline{\F_n}|$,
% and $\prod_{i=0}^{n-1}|\LL_i|  \le ({2^L}/{n})^n$,
% we get that $L\ge \frac{1}{n}\lg\frac{|\overline{\F_n}|n^n}{n!}$.
Thus
\[ \frac{|\F_n|}{n!} \;\le\; |\overline{\F_n}| \;\le\; \prod_{i=0}^{n-1}|\LL_i| \;\le\; \left(\frac{2^L}{n}\right)^n\;, \]
or equivalently
\[ L \;\ge\; \frac{1}{n}\lg\frac{|\overline{\F_n}|n^n}{n!}
\;=\; \frac{1}{n}\lg|\F_n| + \frac{1}{n}\lg\frac{n^n}{n!} \;. \]
It is easy to verify that $\frac{1}{n}\lg\frac{n^n}{n!}$ is increasing in~$n$ and tends to $\lg{\rm e}=1.41695\ldots$ as $n\to\infty$. (By Stirling's formula, $\frac{1}{n}\lg\frac{n^n}{n!} \sim \lg{\rm e}-\frac{\lg\sqrt{2\pi n}}{n}$.)  It is also easy to verify that $\frac{1}{n}\lg\frac{n^n}{n!}> 1.4$ for $n\ge 200$.
\end{proof}

For directed graphs we have $\lg|\F_n|=n(n-1)$. For undirected graphs and tournaments we have $\lg|\F_n|={n\choose 2}$. Using Theorem~\ref{thm:lower1} and Theorem~\ref{thm:lower2} we get:

\begin{corollary}\label{C-directed} If there is an $L$-bit adjacency labeling scheme for $n$-vertex \emph{directed} graphs, then $L\ge n$. If the labeling scheme is indexing, then $L\ge n+1$.
\end{corollary}

\begin{corollary}\label{C-undirected} If there is an $L$-bit adjacency labeling scheme for $n$-vertex \emph{undirected} graphs or for $n$-vertex \emph{tournaments}, then $L\ge \lceil\frac{n}{2}\rceil$. If the labeling scheme is indexing, then $L\ge \lceil\frac{n}{2}\rceil+1$.
\end{corollary}

%\begin{proof} Let $\G_n$ be the family of undirected graphs on $[n]$. Clearly, $|\G_n|=2^{{n \choose 2}}$.
%Thus, by Theorem~\ref{T-lower} we have $L\ge {n \choose 2}/n = \frac{n-1}{2}$. As $L$ is an integer, we have $L\ge \lceil\frac{n-1}{2}\rceil$. If labels are distinct, then $L>{n \choose 2}/n = \frac{n-1}{2}$, and thus $L\ge \lceil\frac{n}{2}\rceil$.
%\end{proof}
%
%A \emph{tournament} is a directed graph $G=(V,E)$ such that for every $u\ne v\in V$ we either have $(u,v)\in E$, or $(v,u)\in E$, but not both. As the number of directed graphs, without self-loops, on $[n]$ is $|\DG_n|=2^{n(n-1)}$, and the number of \emph{tournaments} on~$[n]$ is $|{\cal T}_n| = 2^{{n \choose 2}}$, we get:
%
%
%\begin{corollary}\label{C-tournament} Any $L$-bit adjacency labeling scheme for $n$-vertex \emph{tournaments} must $L\ge \lceil\frac{n-1}{2}\rceil$. If labels are required to be distinct, then $L\ge \lceil\frac{n}{2}\rceil$.
%\end{corollary}

Using a slightly more tedious counting we get the following lower bound for bipartite graphs.
% The proof can be found in the appendix.

\begin{corollary}\label{C-bipartite}
If there is an $L$-bit adjacency labeling scheme for $n$-vertex \emph{bipartite} graphs, then $L\ge \lceil\frac{n}{4}\rceil$. If the labeling scheme is indexing, then $L\ge \lceil\frac{n}{4}\rceil+1$.
\end{corollary}

%\begin{proof} If $n$ is even, then the number of named bipartite graphs on~$n$ vertices is clearly at least $2^{(n/2)^2}=2^{(n^2)/4}$, just by considering bipartite graphs with the bipartition $\{0,1,\ldots,n/2-1\}$, and $\{n/2,n/2+1,\ldots,n-1\}$. If~$n$ is odd, the number is at least $2^{(n-1)(n+1)/4}=2^{(n^2-1)/4}$. The lower bound follows.
%\end{proof}
%
%Counting the number of (named) bipartite graphs exactly is not an easy task. However, it is easy to see that their number is \emph{at most}
%\[\sum_{k=1}^n {n \choose k} 2^{k(n-k)} \;\le\;
%% 2^n(2^{(n^2)/4}+2^{(n^2-1)/4} + 2^{(n^2-4)/4} + \ldots \]
%2^n 2^{(n^2)/4} \sum_{i=0}^\infty 2^{-i^2/4} \;<\; 2^{(n^2)/4+n+3/2}\;. \]
%Thus, Theorem~\ref{T-lower} cannot be used to improve the bound of Corollary~\ref{C-bipartite} beyond $\frac{n}{4}+2$. A slightly more careful counting in the proof of Corollary~\ref{C-bipartite} can be used to improve the bound to $\frac{n}{4}+1$.

\section{Concluding remarks}\label{sec:concl}

We presented improved adjacency labeling schemes for directed, undirected and bipartite graphs. Our schemes are almost optimal. They give rise to almost optimal induced-universal graphs for these families of graphs. We also presented slightly improved lower bounds. Closing the small remaining gaps between our upper and lower bounds is an interesting open problem.

An \emph{oriented} graph is a directed graph with no anti-parallel edges. We believe that using our techniques it is also possible to design an $(\frac{\lg 3}{2}n+O(1))$-bit adjacency labeling scheme for $n$-vertex oriented graphs. We also believe that the techniques we used for bipartite graphs could also be used to design almost optimal schemes for other \emph{hereditary} families of graphs. (For more on hereditay families of graphs see Bollob{\'a}s and Thomason \cite{countbitp}.)

\bibliographystyle{plain}
\bibliography{labels}

\makeatletter
\def\runninghead{\hrulefill\quad APPENDIX\quad\hrulefill}
\def\ps@headings{
\def\@oddhead{\footnotesize\rm\hfill\runninghead\hfill}}
\def\@evenhead{\@oddhead}
\def\@oddfoot{\rm\hfill\thepage\hfill}\def\@evenfoot{\@oddfoot}
\makeatother

\newpage
\setlength{\headsep}{15pt} \pagestyle{headings}

\appendix

\section{A modified spreading lemma}\label{sec:spread2}

% Using this more efficient encoding of the indices we get:

%\begin{lemma}\label{lem:spread2} Suppose that $n=2^b+c$, where $0\le c<2^b$.
%Let $0\le \ell_i\le n-k$, for $i\in [k]$. Then,
%there is a labeling scheme for $(k,n-k)$-bipartite graphs that assigns vertex $u_i$, for $i\in[k]$, a $(\lceil\lg n\rceil+(n-k)-\ell_i)$-bit label, and assigns each vertex $v_j$, for $j\in[n-k]$, an $(\lceil\lg n\rceil+L)$-bit label, where
%$L=\lceil (({\sum_{i=0}^{k-1} \ell_i})+2c)/{(n-k)} \rceil-1$. The first $\lfloor\lg n\rfloor$ or $\lceil\lg n\rceil$ bits in the labels of the vertices form a prefix-free encoding of distinct indices that may be chosen arbitrarily.
%%$L=\left\lceil \frac{\sum_{i=0}^{k-1} \ell_i}{k_2} \right\rceil$.
%\end{lemma}

It is sometimes useful to have the spreading lemma assign tags of slightly different lengths to the vertices of~$V$. The following version receives an additional parameter $0\le C\le n-k$. Vertices of~$V$ of index smaller than~$C$ are assigned $L$-bit tags, while those with index at least~$C$ are assigned $(L+1)$-bit tags. This difference is later used to offset the difference in the number of bits needed to encode each index.

\begin{lemma} \label{lem:spread2}\mbox{\rm [Spreading]} For every $0\le \ell_i\le n-k$, where $i\in [k]$, and every $0\le C\le n-k$, there is a labeling scheme with the following properties. The scheme receives an $(k,n-k)$-bipartite graph $G=(U,V,E)$, where $|U|=k$, $|V|=n-k$, with a distinct index $\,ind_1(u)\in[k]$ assigned to every vertex $u\in U$ and a distinct index $\,ind_2(v)\in [n-k]$ assigned to every vertex $v\in V$. The scheme assigns each vertex $u\in U$ an $((n-k)-\ell_i)$-bit tag $adj_1(u)$, where $i=ind_1(u)$. It assigns each vertex $v\in V$ a tag $adj_2(v)$. If $ind_2(v)\in [0,C)$, then $adj_2(v)$ is of length $L$, otherwise it is of length $L+1$, where $L=\lceil (({\sum_{i=0}^{k-1} \ell_i})+C)/{(n-k)} \rceil-1$.
For every $u\in U$ and $v\in V$, given $(ind_1(u), adj_1(u))$ and $(ind_2(v),adj_2(v))$, and given the $\ell_i$'s, it is possible to determine whether $(u,v)\in E$.
\end{lemma}

\begin{proof}
The proof is almost identical to the proof of Lemma~\ref{lem:spread}. The only difference is that we start spreading the bits of~$U$ to the vertices of~$V$ starting with the vertex of index~$C$.
This is easily achieved by letting $s_0=C$, and $s_i=(s_{i-1}+\ell_i)\bmod(n-k)$, for $i>0$. After moving the first $(n-k)-C$ bits from vertices of~$U$, each vertex of index at least~$C$ gets exactly one bit, and  only $(\sum_{i=0}^k\ell_i)-((n-k)-C)$ additional bits need to be spread among the vertices of~$V$. Each vertex of~$V$ gets only
\[ L \;=\;
\left\lceil \frac{({\sum_{i=0}^{k-1} \ell_i})-((n-k)-C)}{n-k} \right\rceil \;=\;
\left\lceil \frac{({\sum_{i=0}^{k-1} \ell_i})+C}{n-k} \right\rceil-1 \]
additional bits.
\end{proof}

\section{An slightly improved encoding of indices}\label{sec:indices}

When~$n$ is not a power of~$2$, and especially when~$n$ is just slightly larger than a power of~$2$, using $\lceil \lg n\rceil$ bits to represent each index is a bit wasteful (pun intended).
% \footnote{We do not actually loose a whole bit in the encoding of the index of each vertex, just a fraction of a bit, but this is sometimes significant, as we shall see below.}
A slightly more economical encoding can be used.
% As the improved encoding is only used to fine-tune our constructions, it may be skipped at first reading.

Suppose that $n=2^{b-1}+c$, where $0< c\le 2^{b-1}$. Note that $b=\lceil\lg n\rceil$. If $0\le i<2c$, we encode~$i$ using the $b$-bit binary representation of~$i$. If $2c\le i<n$, we encode it using the $(b-1)$-bit binary representation of~$i-c$. For example, if $n=5=2^2+1$, then $b=3$, $c=1$, and the encoding of the indices are $000,001,01,10,11$. It is easy to check that this is a prefix-free encoding.
If the first~$b-1$ bits of an index describe a number less than~$c$, the next bit is also part of the index, otherwise it is not.
% Given the first~$b$ bits of an index, we can easily decide whether the following bit is also a part of the index or not.

\section{An improved scheme for directed graphs}\label{sec:directed2}

\begin{theorem}\label{thm:directed2} For any $n\ge 100$, there is a labeling scheme for $n$-vertex directed graphs that assigns each vertex an $(n+3)$-bit label.
\end{theorem}

\begin{proof} Suppose that $n=\beta 2^b$ where $b=\lceil\lg n\rceil$ and $\frac12<\beta\le 1$. Note that $n=2^{b-1}+c$ where $c=(\beta-\frac12)2^b$ and thus $2c/n=(2\beta-1)/\beta$. We repeat the proof of Theorem~\ref{thm:directed} using the more economical way of encoding indices described in Appendix~\ref{sec:indices}, and using Lemma~\ref{lem:spread2}, with $C=2c$, instead of Lemma~\ref{lem:spread}. Note that the vertices for which we need one more bit to encode their index are exactly those that get one less bit by the modified spreading lemma.
Each vertex of~$B$ thus gets a label composed of $n+1+\Delta$ bits, where
\[ \Delta \;\le\;
\left\lceil \frac{k(k+1) + \bar{H}(2^{k-1}/n)n + 2c}{n-k} \right\rceil -1 \;=\;
\left\lceil \frac{k(k+1)}{n-k} + \frac{n}{n-k}\left( \bar{H}(\frac{1}{8\beta}) + \frac{2\beta-1}{\beta}\right) \right\rceil -1 \;. \]
It is not difficult to verify that $f(\beta)=\bar{H}(\frac{1}{8\beta}) + \frac{2\beta-1}{\beta}$ is an increasing function of~$\beta$, when $\beta\in (\frac{1}{2},1]$ and that $f(1)=\bar{H}(1/8)+1<2.346$. It is not difficult to verify that for $n\ge 100$ we have $\frac{k(k+1)}{n-k}<0.5$ and $\frac{n}{n-k}f(1)<2.5$, and thus $\Delta\le 2$.

Thus, the label of each vertex of~$B$ contains at most $n+3$ bits. The labels of the vertices of~$A$ contain only $n+2$ bits, as before, and are padded to length $n+3$.
\end{proof}

\section{An improved scheme for undirected graphs}\label{sec:undirected2}

\begin{theorem}\label{thm:undirected2} For any $n\ge 100$, there is a labeling scheme for $n$-vertex undirected graphs that assigns each vertex an $(\lceil\frac{n}{2}\rceil+4)$-bit label.
\end{theorem}

\begin{proof}
We begin by proving that the claim for odd values of $n$.
% Note that for odd $n$ we have $\lceil\frac{n}{2}\rceil+4=\frac{n-1}{2}+5$.
We use the same approach used in the proof of Theorem~\ref{thm:directed2}. Suppose that $n=\beta 2^b$ where $b=\lceil\lg n\rceil$ and $\frac12<\beta\le 1$. We again have $n=2^{b-1}+c$ where $c=(\beta-\frac12)2^b$ and thus $2c/n=(2\beta-1)/\beta$. Using the slightly more efficient technique to code the indices, and Lemma~\ref{lem:spread2}, this time with $C=c$, we get that each vertex of $B_0\cup B_1$ is assigned a label of size at most $\frac{n-1}{2}+3+\Delta$, where
\[ \Delta \;=\;
% \left\lceil \frac{\sum_{i=0}^{k-1}(k+\ell_i)}{\lfloor \frac{n}{2}\rfloor-k} \right\rceil \;\le\;
% \left\lceil \frac{k(k+2) + \frac{n}{2}\sum_{i=0}^{k-1} H(2^{i+1}/n)}{\lfloor \frac{n}{2}\rfloor-k} \right\rceil \;\le\;
\left\lceil \frac{k(k+2)}{\frac{n-1}{2}-k} + \frac{\frac{n}{2}}{{ \frac{n-1}{2}-k}}
%\bar{H}(2^{k}/n)
\left( \bar{H}(\frac{1}{8\beta}) + \frac{2\beta-1}{\beta}\right)
\right\rceil - 1 \;,
\]
with the familiar function $f(\beta)=\bar{H}(\frac{1}{8\beta}) + \frac{2\beta-1}{\beta}$ appearing again.
Again $\Delta\le 2$, and thus the number of bits in each label is at most $\frac{n-1}{2}+5 = \lceil\frac{n}{2}\rceil+4$.

We now turn to the case where~$n$ is even. In the proof of Theorem~\ref{thm:undirected}, the tag that each vertex of $B_0\cup B_1$ is assigned by Moon's scheme, used to represent $G[B_0\cup B_1]$, is of length $\frac{n}{2}-k$. As mentioned after the proof of Theorem~\ref{thm:moon}, this is somewhat wasteful, as $\frac{n}{2}-k$ adjacency bits are actually stored twice. We can thus remove these redundant bits from the tags, saving on average half a bit for each vertex. More precisely, half of the tags would now be of length $\frac{n}{2}-k-1$ and half of length $\frac{n}{2}-k$. We now use a further modified version of Lemma~\ref{lem:spread2} to do the spreading. We start moving bits to the vertices whose tags are of length $\frac{n}{2}-k-1$. It is not difficult to check that the label of each vertex of $B_0\cup B_1$ would now be of length at most $\frac{n}{2}+3+\Delta'$, where
\[ \Delta' \;=\;
% \left\lceil \frac{\sum_{i=0}^{k-1}(k+\ell_i)}{\lfloor \frac{n}{2}\rfloor-k} \right\rceil \;\le\;
% \left\lceil \frac{k(k+2) + \frac{n}{2}\sum_{i=0}^{k-1} H(2^{i+1}/n)}{\lfloor \frac{n}{2}\rfloor-k} \right\rceil \;\le\;
\left\lceil \frac{k(k+2)}{\frac{n}{2}-k} + \frac{\frac{n}{2}}{{ \frac{n}{2}-k}}
%\bar{H}(2^{k}/n)
\left( \bar{H}(\frac{1}{8\beta}) + \frac{2\beta-1}{\beta}-\frac{1}{2}\right)
\right\rceil - 1 \;.
\]
Let $g(\beta)=\bar{H}(\frac{1}{8\beta}) + \frac{2\beta-1}{\beta}-\frac{1}{2}$. It is not difficult to check that for $\beta\in (\frac12,1]$ we have $g(\beta)<2$. For sufficiently large~$n$ we thus $\Delta'= 1$ and the claim of the Theorem follows.
\end{proof}

\section{Bipartite graphs}\label{sec:bipartite2}

\begin{proof}(Of Theorem~\ref{thm:biased}) To represent an $(\frac{n}{2}-r,\frac{n}{2}+r)$ bipartite graph we need $(\frac{n}{2}-r)(\frac{n}{2}+r)=\frac{n^2}{4}-r^2$ bits. Using the spreading lemma we can split these bits almost evenly among the vertices, giving each vertex a tag of
$\lceil \frac{n}{4}-\frac{r^2}{n}\rceil$ bits.
\end{proof}

\begin{proof}(Of Theorem~\ref{thm:balanced}) The proof is similar to the proofs of Theorems~\ref{thm:directed} and~\ref{thm:undirected}, though the amount of details increases yet again. Let $G=(U,V,E)$ be an $(\frac{n}{2},\frac{n}{2})$ bipartite graph. We split $U$ into four sets $A_{0,0},B_{0,0},A_{1,0},A_{1,0}$ of sizes $k$, $\ceil{\frac{n}{4}}-k$, $k$ and $\floor{\frac{n}{4}}-k$, respectively, where $k=\ceil{\lg n}-4$. We similarly split $V$ into four sets $A_{0,1},B_{0,1},A_{1,1},A_{1,1}$. We now use Lemma~\ref{lem:bipartite} to assign tags to the graphs $G[A_{0,0},B_{0,1}]$, $G[A_{1,0},B_{1,1}]$, $G[A_{0,1},B_{0,0}]$, $G[A_{1,1},B_{0,1}]$ and use the spreading lemma to assign tags to the remaining subgraphs. Using calculations similar to the ones made in the proofs of Theorems~\ref{thm:directed} and~\ref{thm:undirected}, we get the claimed result.
\end{proof}

\end{document}